%% file: main.tex
\documentclass[a4paper, autoref, numberwithinsect, USenglish]{article}
\usepackage[utf8]{inputenc}
\usepackage{microtype}
\usepackage{authblk}
\usepackage{hyperref}

\usepackage{xspace}
\usepackage{amsmath}
\usepackage{amsfonts}
\usepackage{amssymb}
\usepackage{amsthm}
\usepackage{cleveref}
\usepackage{microtype}
\usepackage{wrapfig}
\usepackage{thmtools,thm-restate}
\usepackage[textsize=small]{todonotes}

\input{images/images-include.tex}
\input{images/style.tikz}
\input{images/common-picture-macros.tikz}

\input{macros}

\newtheorem{lemma}{Lemma}
\newtheorem{theorem}{Theorem}
\newtheorem{claim}{Claim}

\date{\vspace{-5ex}}

\title{New Pumping Technique for 2-dimensional VASS}

\author[1]{\href{wczerwin@mimuw.edu.pl}{Wojciech Czerwi\'nski}}

\author[2]{\href{sl@mimuw.edu.pl}{S{\l}awomir Lasota}\footnote{Partially supported by the Polish NCN grant 2017/27/B/ST6/02093}}

\author[3]{\href{loeding@informatik.rwth-aachen.de}{Christof L\"{o}ding}}

\author[4]{\href{r.piorkowski@mimuw.edu.pl}{Rados{\l}aw Pi{\'o}rkowski}\footnote{Partially supported by the Polish NCN grant 2016/21/B/ST6/01505}}

\affil[1,2,4]{University of Warsaw}
\affil[3]{RWTH Aachen}

\begin{document}

\maketitle

\begin{abstract}
We propose a new pumping technique for 2-dimensional vector addition systems with states (2-VASS)
building on natural geometric properties of runs.
We illustrate its applicability by reproving an exponential bound on the length of the shortest accepting run,
and by proving a new pumping lemma for languages of 2-VASS.
The technique is expected to be useful for settling questions concerning languages of 
2-VASS, e.g., for establishing decidability status of the regular separability problem. 
\end{abstract}

\input{intro}

\input{prelim}

\input{thin-thick}

\input{proof-reach}

\newpage

\bibliographystyle{plain}
\bibliography{main}

\newpage

\appendix

\input{first-cycle}

\input{proof-reach-thin}

\end{document}

%% file: images/images-include.tex
\usepackage{tikz}
\usetikzlibrary{calc,fadings,arrows.meta,positioning,decorations.pathmorphing}
\pgfdeclarelayer{background}
\pgfdeclarelayer{foreground}
\pgfdeclarelayer{edgelayer}
\pgfdeclarelayer{nodelayer}
\pgfsetlayers{background,edgelayer,nodelayer,main,foreground}

%% file: images/style.tikz
\tikzstyle{basicArrow}=[draw=black,thick,arrows={-Triangle}]
\tikzstyle{basicArrowRev}=[draw=black,thick,arrows={Triangle-}]
\tikzstyle{basicArrowBoth}=[draw=black,thick,arrows={Triangle-Triangle}]

\definecolor{cMain}{RGB}{83,188,166}
\definecolor{cA}{RGB}{83,188,166}
\definecolor{cB}{RGB}{194,252,115}

\tikzstyle{none}=[]
\tikzstyle{vector}=[basicArrow]
\tikzstyle{vectorShadeA}=[vector,draw=black,thin,arrows={-Triangle[width=3pt, length=3pt]}]
\tikzstyle{vectorShadeB}=[vector,draw=black,thin,arrows={-Triangle[width=3pt, length=3pt]}]
\tikzstyle{vectorLabel}=[auto,midway,inner sep=0.04cm]

\tikzstyle{point}=[circle,fill=black,inner sep=0.3ex]
\tikzstyle{smallPoint}=[circle,fill=black,inner sep=0.15ex]
\tikzstyle{runPoint}=[circle,fill={rgb,255: red,40; green,40; blue,40},inner sep=0.15ex,fill opacity=1]
\tikzstyle{runLine}=[draw=cMain,draw opacity=1,thick]

\tikzstyle{runOutlineA}=[draw=cA,draw opacity=0.5,line width=4mm]
\tikzstyle{runOutlineB}=[draw=cB,draw opacity=0.5,line width=4mm]
\tikzstyle{runLineBlack}=[draw=black]

\tikzstyle{beltBg}=[fill=cMain,fill opacity=0.3]
\tikzstyle{accentLine}=[draw={rgb,255: red,131; green,217; blue,71}]

\tikzstyle{sequentialConeShade}=[fill=cMain, fill opacity=0.5]

\tikzstyle{sequentialConeShadeA}=[fill=cA, fill opacity=1]
\tikzstyle{sequentialConeShadeB}=[fill=cB, fill opacity=1]

\tikzstyle{sequentialConeAngleA}=[fill=cA!20,draw=black,thick, fill opacity=1]
\tikzstyle{sequentialConeAngleB}=[fill=cB!20,draw=black,thick, fill opacity=1]

%% file: images/common-picture-macros.tikz
\newcommand{\drawBelt}[4]{ 
    \coordinate (vector) at (#1,#2);
    \coordinate (beltStart) at (0,0);
    \coordinate (beltEnd) at ($ (beltStart)!2*#4!0:(vector) $);
    \coordinate (beltLeftStart) at ($ (beltStart)!#3!90:(beltEnd) $);
    \coordinate (beltRightStart) at ($ (beltStart)!#3!270:(beltEnd) $);
    \coordinate (beltLeftEnd) at ($ (beltEnd)!#3!270:(beltStart) $);
    \coordinate (beltRightEnd) at ($ (beltEnd)!#3!90:(beltStart) $);

    \draw[dotted] (beltStart) -- (beltEnd);

    \begin{pgfonlayer}{background}
        \clip (0,0) rectangle (#4,#4);
        \fill[beltBg] (beltLeftStart) -- (beltLeftEnd) -- (beltRightEnd) -- (beltRightStart) -- cycle;
    \end{pgfonlayer}
}

\newcommand{\drawBeltB}[5]{ 
    \coordinate (vector) at (#1,#2);
    \coordinate (beltStart) at (0,0);
    \coordinate (beltEnd) at ($ (beltStart)!2*#4!0:(vector) $);
    \coordinate (beltLeftStart) at ($ (beltStart)!#3!90:(beltEnd) $);
    \coordinate (beltRightStart) at ($ (beltStart)!#3!270:(beltEnd) $);
    \coordinate (beltLeftEnd) at ($ (beltEnd)!#3!270:(beltStart) $);
    \coordinate (beltRightEnd) at ($ (beltEnd)!#3!90:(beltStart) $);

    \draw[dotted] (beltStart) -- (beltEnd);

    \begin{pgfonlayer}{background}
        \clip (0,0) rectangle (#4,#5);
        \fill[beltBg] (beltLeftStart) -- (beltLeftEnd) -- (beltRightEnd) -- (beltRightStart) -- cycle;
    \end{pgfonlayer}
}

\newcommand{\drawGrid}[1]{
    \node (S) at (0, 0) {}; 
    \node (X) at (#1, 0) {}; 
    \node (Y) at (0, #1) {}; 
    
    \begin{pgfonlayer}{foreground}
        \draw[basicArrowBoth] (X.center) to (S.center) to (Y.center);
    \end{pgfonlayer}
}

\newcommand{\drawGridB}[2]{
    \node (S) at (0, 0) {}; 
    \node (X) at (#1, 0) {}; 
    \node (Y) at (0, #2) {}; 
    
    \begin{pgfonlayer}{foreground}
        \draw[basicArrowBoth] (X.center) to (S.center) to (Y.center);
    \end{pgfonlayer}
}

\newcommand{\drawFading}[1]{
    \begin{pgfonlayer}{foreground}
        \fill[fill=white,path fading=south] (0,#1-1) rectangle (#1+0.01,#1+0.01);
        \fill[fill=white,path fading=west] (#1-1,0) rectangle (#1+0.01,#1+0.01);
        \fill[fill=white] (0,#1-0.9) rectangle (#1-0.8,#1-0.8);
        \fill[fill=white] (0,#1-0.6) rectangle (#1-0.5,#1-0.5);
        \fill[fill=white] (0,#1-0.3) rectangle (#1-0.2,#1-0.2);
        
        \fill[fill=white] (#1-0.9,0) rectangle (#1-0.8,#1-0.8);
        \fill[fill=white] (#1-0.6,0) rectangle (#1-0.5,#1-0.5);
        \fill[fill=white] (#1-0.3,0) rectangle (#1-0.2,#1-0.2);
    \end{pgfonlayer}
}

\newcommand{\drawFadingB}[2]{
    \begin{pgfonlayer}{foreground}
        \fill[fill=white,path fading=south] (0,#2-1) rectangle (#1+0.01,#2+0.01);
        \fill[fill=white] (0,#2-0.9) rectangle (#1-0.8,#2-0.8);
        \fill[fill=white] (0,#2-0.6) rectangle (#1-0.5,#2-0.5);
        \fill[fill=white] (0,#2-0.3) rectangle (#1-0.2,#2-0.2);
        
    \end{pgfonlayer}
}

%% file: macros.tex
\newcommand{\ignore}[1]{}

\newcommand{\pspace}{{\sc PSpace}\xspace}
\newcommand{\nl}{{\sc NL}\xspace}
\renewcommand{\a}{\alpha}

\newcommand{\U}{{\cal U}}

\newcommand{\SCC}{{\sc scc}\xspace}

\newcommand{\emptyseq}{\varepsilon}

\newcommand{\belt}[2]{{\cal B}_{#1, #2}}
\newcommand{\size}[1]{|#1|}
\newcommand{\nvass}[1]{#1\text{-VASS}\xspace}
\newcommand{\onevass}{\nvass{1}}
\newcommand{\twovass}{\nvass{2}}
\newcommand{\C}{\mathcal{C}}

\newcommand{\N}{\mathbb{N}}
\newcommand{\Z}{\mathbb{Z}}

\newcommand{\Zpos}{\mathbb{Z}_{> 0}}

\newcommand{\R}{\mathbb{Q}}

\newcommand{\Rnonneg}{\R_{\geq 0}}
\newcommand{\Rpos}{\R_{> 0}}

\newcommand{\set}[1]{\{#1\}}
\newcommand{\setof}[2]{\set{#1 \mid #2}}

\newcommand{\prettyforall}[2]{\forall_{#1} \, #2}

\newcommand{\vass}{VASS\xspace}

\newcommand{\norm}[1]{\lVert#1\rVert}

\newcommand{\src}{\textup{src}}
\newcommand{\trg}{\textup{trg}}
\newcommand{\rev}{\textup{rev}}
\newcommand{\dist}{\textup{dist}}

\newcommand{\clockwise}{\circlearrowright}

\newcommand{\conf}{\textup{Conf}}
\newcommand{\poscone}{\textup{\sc SeqCone}}
\newcommand{\cone}{\textup{\sc Cone}}
\newcommand{\zero}{(0,0)\xspace}
\renewcommand{\angle}[4]{\measuredangle #1 #3, #4 #2}
\newcommand{\ccangle}[2]{\angle{[}{]}{#1}{#2}}
\newcommand{\ocangle}[2]{\angle{(}{]}{#1}{#2}}
\newcommand{\coangle}[2]{\angle{[}{)}{#1}{#2}}
\newcommand{\ooangle}[2]{\angle{(}{)}{#1}{#2}}

\newcommand{\eff}{\textup{eff}}

\renewcommand{\enspace}{}

%% file: intro.tex
\section{Introduction}

Vector addition systems~\cite{KM69} are a widely accepted model of concurrency equivalent to Petri nets.
Another equivalent model, called vector addition systems with states (\vass)~\cite{HP79},
is an extension of finite automata with integer counters, on which
the transitions can perform operations of increment or decrement (but no zero tests), with the proviso that
counter values are 0 initially and must stay non-negative along a run. The number of counters
$d$ defines the \emph{dimension} of a \vass. For brevity, we call a \vass of dimension $d$ a $d$-\vass.
Formally, every transition of a \nvass{$d$} $V$ has adjoined a vector $v\in\Z^d$ describing the effect of executing 
this transition on counter values; thus a transition is a triple $(q, v, q') \in Q\times\Z^d \times Q$, 
where $Q$ is the set of control states of $V$.
A finite \emph{path}, i.e., a sequence of transitions of the form
$ \label{eq:path}
\pi \ = \  (q_0, v_1, q_1), (q_1, v_2, q_2), \ldots, (q_{n-1}, v_n, q_n)
$,
induces a run if the counter values stay non-negative, i.e., $v_1 + \ldots + v_i \in \N^d$ for every $i$.

In this paper we concentrate on \emph{pumping}, i.e., techniques exploiting repetitions of states in runs.
Pumping is an ubiquitous phenomenon which typically provides valuable tools in proving short run properties, 
or showing language inexpressibility results. 
It seems to be particularly relevant in case of \vass, as even
the core of the seminal decision procedure for the reachability problem in \vass
by Mayr and Kosaraju~\cite{DBLP:conf/stoc/Mayr81, DBLP:conf/stoc/Kosaraju82} is fundamentally 
based on pumping:
briefly speaking, the decision procedure decomposes a \vass into a finite number of \vass, each of them admitting a property that
every path can be pumped up so that it induces a run.
Pumping techniques are used even more explicitly when dealing with subclasses of \vass of bounded dimension. 
The \pspace upper bound for the reachability problem in \twovass~\cite{DBLP:conf/lics/BlondinFGHM15}
relies on various un-pumping transformations of an original run, leading to
a simple run of at most exponential length, in the form of a short path with adjoined short disjoint cycles.
A smart surgery on those simple runs was 
also used to obtain a stronger upper bound (\nl)
in case when the transition effects are 
represented in unary~\cite{DBLP:conf/lics/EnglertLT16}.
Un-pumping is also used in~\cite{DBLP:conf/fossacs/ChistikovCHPW16} to provide a quadratic bound
on the length of the shortest run for \onevass, also known as one counter automata without zero tests, 
and for unrestricted one counter automata. 
See also~\cite{Hofman16,Latteux83} for pumping techniques in one counter automata.

\subparagraph{Contribution.}

\begin{wrapfigure}[26]{r}[0cm]{5cm}  
  \centering
  \vspace{-\baselineskip}
  \input{images/01-small-thin-run.tikz}
  \input{images/02-small-thick-run.tikz}
  \caption{Thin (above) and thick run (below).
  Points correspond to counter values, 
  and control states along a run are ignored.}
\end{wrapfigure}
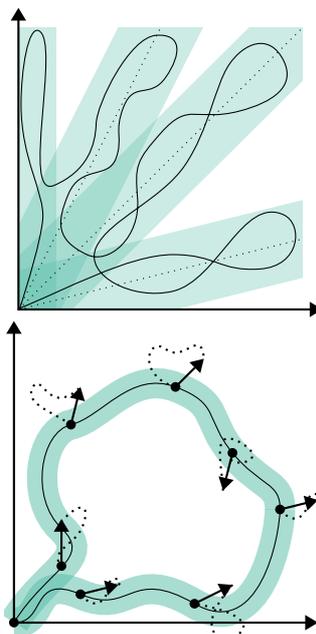

The above-mentioned techniques are mostly oriented towards reachable sets, and henceforth
may ignore certain runs as long as the reachable set is preserved.
In consequence, they are not very helpful in solving decision problems
formulated in terms of the whole language accepted by a \vass, like the regular separability problem
(cf.~the discussion below).
Our primary objective is to design a pumping infrastructure applicable to \emph{every} run of a \twovass.
%
Therefore, as our main technical contribution we perform a thorough 
classification of runs, in the form of a dichotomy (see the illustrations on the right): 
for every run $\pi$ of a \twovass, whose initial and final values of both counters are 0,
%
\begin{itemize}
\item either $\pi$ is \emph{thin}, by which we mean that the counter values along the run stay
within belts, whose direction and width are all bounded polynomially in the number of states and the largest absolute value of vectors of the \twovass;
\item or $\pi$ is \emph{thick}, by which we mean that a number of cycles is enabled along the run, the effect
vectors of these cycles span (slightly oversimplifying) the whole plane, and furthermore the lengths of cycles and the extremal factors of
$\pi$ are all bounded polynomially in $M$ and exponentially in $n$.
(For the sake of simplicity some details are omitted here; 
the fully precise statement of the dichotomy is \cref{lem:run-2d} in \cref{sec:thin-thick}).
\end{itemize}
The dichotomy immediately entails a pumping lemma for \twovass by using, essentially,
the pumping scheme of \onevass in case of thin runs, and the cycles enabled along a run in case of thick runs 
(cf.~\cref{thm:pumping}). 
As a more subtle application of the dichotomy, we derive an alternative proof of the exponential run property
(shown originally in~\cite{DBLP:conf/lics/BlondinFGHM15}), 
which immediately implies \pspace-membership of the reachability problem (cf.~\cref{thm:short-run}).
%

\subparagraph{Further applications.}

We envisage other possible applications of the dichotomy. One important case can be the regular separability problem:
given two \emph{labeled} \twovass $V_1$ and $V_2$, decide if there is a regular language separating languages of
$V_1$ and $V_2$, i.e., including one of them and disjoint from the other.
The problem is decidable in \pspace for \onevass~\cite{DBLP:conf/lics/CzerwinskiL17}
while the decidability status for \twovass is still open.
A cornerstone of the decision procedure of~\cite{DBLP:conf/lics/CzerwinskiL17} is a well-behaved over-approximation
of a language of a \onevass $V$ by a sequence of regular languages $(V_n)_{n\in\N}$, where the precision of
approximation increases with increasing $n$. In case of \onevass, the language $V_n$ is obtained by abstraction
of $V$ \emph{modulo $n$}; on the other hand, as argued in~\cite{DBLP:conf/lics/CzerwinskiL17},
the very same approach necessarily fails for dimensions larger than 1. 
It seems that our dichotomy classification of runs of a \twovass prepares the ground for the right definition of
abstraction $V_n$ modulo $n$. Indeed, intuitively speaking, as long as the run stays within belts, 1-dimensional
counting modulo $n$ along the direction of a belt is sufficient; 
otherwise, a 2-dimensional abstraction modulo $n$ can be applied as soon as a sufficient number of pumpable 
cycles has been identified along a run. 

As our approach builds on natural geometric properties of runs,
we believe that it can be generalized to dimensions larger than 2. However, one should not expect efficient length bounds
from this generalization itself,
as already in dimension 3 
the prefix of a run preceding the first pumpable cycle has non-elementary length
(the length can be as large as tower of $n$ exponentials in the composition of $n$ copies of the
Hopcroft and Pansiot example~\cite{HP79}).

%


%% file: images/01-small-thin-run.tikz
\begin{tikzpicture}[scale=0.25]
    \newcommand{\gridSize}{16}
    \drawFading{\gridSize}
    \drawGrid{\gridSize}
    \foreach \vx/\vy in {0/4,2/4,3/3,4/1} {
        \begin{scope}
            \clip (0,0) rectangle (\gridSize,\gridSize);
            \drawBelt{\vx}{\vy}{2cm}{\gridSize}
        \end{scope}
    }


    \begin{scope}[y=1.97, x=1.97, yshift=0, inner sep=0pt, outer sep=0pt]
        \path[runLineBlack] (1.0000,1.0000) .. controls (6.3333,17.1667) and
        (17.3000,53.5000) .. (18.5000,69.5000) .. controls (20.0000,89.5000) and
        (2.0000,114.5000) .. (2.0000,157.0000) .. controls (2.0000,184.0000) and
        (7.5000,217.5000) .. (15.5000,213.5000) .. controls (32.5000,205.0000) and
        (1.5000,84.0000) .. (27.0000,95.0000) .. controls (52.5000,106.0000) and
        (58.5000,126.0000) .. (59.0000,142.0000) .. controls (59.5000,158.0000) and
        (84.5000,208.5000) .. (104.5000,210.0000) .. controls (124.5000,211.5000) and
        (124.0000,195.5000) .. (112.0000,178.0000) .. controls (100.0000,160.5000) and
        (88.0000,175.5000) .. (80.0000,155.5000) .. controls (72.0000,135.5000) and
        (80.0000,131.5000) .. (74.0000,115.5000) .. controls (68.0000,99.5000) and
        (47.5000,111.0000) .. (38.0000,88.5000) .. controls (28.5000,66.0000) and
        (29.0000,53.0000) .. (52.0000,42.0000) .. controls (68.9556,33.8908) and
        (92.8600,55.5000) .. (90.5000,85.0000) .. controls (88.5000,110.0000) and
        (89.0000,113.5000) .. (108.5000,135.5000) .. controls (130.1210,159.8930) and
        (134.5000,143.1840) .. (168.0000,152.0000) .. controls (193.5000,158.7100) and
        (211.0000,171.0000) .. (201.5000,190.5000) .. controls (190.7610,212.5440) and
        (171.0000,204.0000) .. (160.0000,187.5000) .. controls (149.0000,171.0000) and
        (145.0000,138.0000) .. (126.5000,110.5000) .. controls (108.4020,83.5979) and
        (68.3155,68.5824) .. (59.0000,47.5000) .. controls (49.5000,26.0000) and
        (87.0000,6.5000) .. (117.5000,14.0000) .. controls (151.6140,22.3888) and
        (170.5000,86.9667) .. (201.5000,72.5000) .. controls (216.5000,65.5000) and
        (214.5000,33.1379) .. (183.5000,31.0000) .. controls (152.5000,28.8620) and
        (150.6280,42.0000) .. (117.5000,42.0000) .. controls (83.5000,42.0000) and
        (34.1667,14.1667) .. (1.0000,1.0000) -- cycle;
    \end{scope}

\end{tikzpicture}

%% file: images/02-small-thick-run.tikz
\begin{tikzpicture}[scale=0.25]
    \newcommand{\gridSize}{16}
    \clip (-0.5,-1.2) rectangle (\gridSize+0.5,\gridSize);
    \drawGrid{\gridSize}


    \begin{scope}[yshift=16cm]
    \begin{scope}[xscale=0.05,yscale=-0.05]
        \begin{scope}
            \draw[runOutlineA]    (280,200) .. controls (280.5,171) and (259.5,176) .. (230,140) .. controls (200.5,104) and (217.5,88) .. (170,70) .. controls (122.5,52) and (92.5,102) .. node[auto,pos=0.4,swap,inner sep=0]{\color{cA!70!black}} (60,110) .. controls (27.5,118) and (16.5,188) .. (50,220) .. controls (83.5,252) and (34.5,264) .. (0,320) ;
            \draw[runOutlineA]    (0,320) .. controls (34.5,321) and (16.5,267) .. (70,290) .. controls (123.5,313) and (121.5,260) .. (190,300) .. controls (258.5,340) and (279.5,229) .. node[auto,pos=0.5,swap,inner sep=0]{\color{cA!50!black}} (280,200);
        \end{scope}



        \draw[vector]    (190,300) --   (230,280) ;
        \draw[vector]    (50,260)  --   (50,210) ;
        \draw[vector]    (60,110)  --   (70,70) ;
        \draw[vector]    (230,140) --   (220,180) ;
        \draw[vector]    (170,70)  --   (200,40) ;
        \draw[vector]    (70,290)  --   (110,280) ;
        \draw[vector]    (280,200)  --   (320,190) ;

        


        
        \draw [dotted,thick,color=black,draw opacity=1 ]   (230,280) .. controls (191.5,324) and (211.5,346-5) .. (230,340-5) .. controls (248.5,334-5) and (251.5,311) .. (190,300) ;
        \draw [dotted,thick,color=black,draw opacity=1 ]   (110,280) .. controls (116.5,274) and (100.5,269) .. (100,280) .. controls (99.5,291) and (88.5,312) .. (70,290) ;
        \draw [dotted,thick,color=black,draw opacity=1 ]   (320,190) .. controls (326.5,184) and (310.5,179) .. (310,190) .. controls (309.5,201) and (298.5,222) .. (280,200) ;

        \draw [dotted,thick,color=black,draw opacity=1,yshift=42.5cm ]   (50,220) .. controls (37.5,200) and (60.5,197) .. (70,180) .. controls (79.5,163) and (78.5,140) .. (50,170) ;
        \draw [dotted,thick,color=black,draw opacity=1 ]   (60,110) .. controls (37.5,95) and (10.5,87) .. (20,70) .. controls (29.5,53) and (36.5,103) .. (70,70) ;
        \draw [dotted,thick,color=black,draw opacity=1 ]   (170,70) .. controls (128.5,58) and (140.5,15) .. (160,30) .. controls (179.5,45) and (196.5,6) .. (200,40) ;
        \draw [dotted,thick,color=black,draw opacity=1 ]   (220,180) .. controls (219.5,154) and (205.5,109) .. (240,130) .. controls (274.5,151) and (216.5,158) .. (230,140) ;
        
        \draw[runLineBlack]    (0,320) node[point]{} .. controls (34.5,321) and (16.5,267) .. (70,290) node[point]{} .. controls (123.5,313) and (121.5,260) .. (190,300) node[point]{} .. controls (258.5,340) and (279.5,229) .. (280,200);
        \draw[runLineBlack]    (280,200) node[point]{} .. controls (280.5,171) and (259.5,176) .. (230,140) node[point]{} .. controls (200.5,104) and (217.5,88) .. (170,70) node[point]{} .. controls (122.5,52) and (92.5,102) .. (60,110) node[point]{} .. controls (27.5,118) and (16.5,188) .. (50,220) node[point,yshift=-0.5cm]{} .. controls (83.5,252) and (34.5,264) .. (0,320) ;

    \end{scope}
    \end{scope}
\end{tikzpicture}

%% file: prelim.tex

\section{Preliminaries}  

\label{sem:prelim}

\subparagraph{2-dimensional vector addition systems with states.}  

We use standard symbols $\R, \Z, \N$ for the sets of rationals, integers, and non-negative integers, respectively.
Whenever convenient we use subscripts to specify subsets, e.g.,
$\Rnonneg$ for non-negative rationals.
We refer to elements of $\Z^2$  briefly as \emph{vectors}. 
\emph{Non-negative} vectors are elements of $\N^2$, and \emph{positive} vectors are elements
of $\Zpos^2$.
A vector with only non-negative coordinates and at least one positive coordinate is called \emph{semi-positive}; 
it is either positive, or \emph{vertical} of the form $(0, a)$, or \emph{horizontal} of the form $(a, 0)$, 
for $a\in\Zpos$. 


A $2$-dimensional \emph{vector addition system with states} (\twovass)
$V$ consists of a finite set of control states $Q$
and a finite set of transitions $T \subseteq Q \times \Z^2 \times Q$.  
We refer to the vector $v$ as the \emph{effect} of a transition $(p, v, q)$.
A \emph{path} in $V$ from control state $p$ to $q$ is a sequence of transitions 
$\pi = (q_0, v_1, q_1), (q_1, v_2, q_2), \ldots, (q_{n-1}, v_n, q_n) \in T^*$
where $p = q_0$ and $q = q_n$;
it is called a \emph{cycle} whenever the starting and ending control states coincide ($q_0 = q_n$).
The \emph{effect} of a path is defined as $\eff(\pi) = v_1 + \ldots + v_n \in \Z^2$,
and its \emph{length} is $n$.
A cycle is called non-negative, semi-positive or positive, if its effect is so.

A \emph{configuration} of $V$ is an element of $\conf = Q \times \N^2$. 
A transition $t = (p, v, q)$ is \emph{enabled} in a configuration $c = (p', u)$ if $p = p'$ and $u + v \in \N^2$.
Analogously, a path $\pi$ 
is enabled in a configuration $c = (p', u)$ if $q_0 = p'$  and
$u_i = u + v_1 + \ldots + v_i \in \N^2$ for every $i$.
In such case we say that $\pi$ induces a \emph{run} of the form
\[
\rho = (c_0, t_1, c_1), (c_1, t_2, c_2), \ldots, (c_{n-1}, t_n, c_n) \in (\conf \times T \times \conf)^*
\]
with intermediate configurations $c_i = (q_i, u_i)$, from the \emph{source} configuration 
$\src(\rho) = c_0$ to the \emph{target} one $\trg(\rho) = c_n$.
If the source configuration $c_0$ is clear from the context, we do not distinguish between a path enabled in
$c_0$ and a run with source $c_0$, and simply say that the path \emph{is} the run. 
A \emph{$\zero$-run} is a run whose source and target are \emph{$\zero$-configurations}, i.e., 
a configuration whose vector is $\zero$.



We will sometimes relax the non-negativeness requirement on some coordinates:
For $j\in\{1,2\}$, we say that a path $\pi$ is \emph{$\{j\}$-enabled} in a configuration $c = (p', u)$ 
if $q_0 = p'$  and $(u + v_1 + \ldots + v_i)[j] \in \N$ for every $i$.
We also say that $\pi$ is \emph{$\emptyset$-enabled} in $c$ if just $q_0 = p'$. 


The \emph{reversal} of a \twovass $V = (Q,T)$, denoted $\rev(V)$, is a \twovass with the same control states and with transitions
$\setof{(q,-v, p)}{(p,v,q)\in T}$. 
We sometimes speak of the reversal $\rev(\rho)$ of a run $\rho$ of $V$, implicitly meaning a run in the reversal of $V$.

As the \emph{norm} of $v = (v_1, v_2) \in \R^2$, we take the largest of absolute values of $v_1$ and $v_2$,
$
\norm{v} := \max \{|v_1|, |v_2|\}.
$
By the norm of a configuration $c = (q, v)$ we mean the norm of its vector $v$, and
by the norm $\norm{V}$ of a \twovass $V$ we mean the largest among norms of effects of transitions. 




\subparagraph*{Sequential cones.}

\begin{wrapfigure}[16]{r}[0cm]{6.5cm}
  \centering
  \vspace{-\baselineskip}
  \input{images/03-angle-order.tikz}
  \caption{
  Above $u_{1} \clockwise u_{2} \clockwise \dots  \clockwise u_{11} \clockwise u_1$.
  Also, $u_4 \clockwise u_9$, but $u_4 \not\clockwise u_{11}$ and $u_{11} \clockwise u_4$.
  Pairs of vectors $u_i$, $u_{i+6}$ are contralinear, for $i = 1, \ldots, 5$.}\label{fig:clock-order}
\end{wrapfigure}
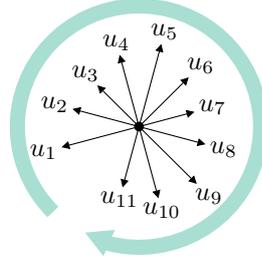

For a vector $v \in \Z^2$, define the half-line induced by $v$ as
$\ell_v := \R_{\geq 0} \cdot v = \{\alpha v \mid \alpha \in \R_{\geq 0}\}$.
We call two vectors $v, w$ \emph{colinear} if $\ell_v = \ell_w$, and \emph{contralinear} if $\ell_v = \ell_{-w}$.
For two vectors $u, v \in \Z^2 \setminus \{\zero\}$, 
define the \emph{angle} $\ccangle{u}{v}\subseteq \R^2$ as  the union 
of all half-lines which lie clock-wise between $\ell_u$ and $\ell_v$, including the two half-lines themselves.
In particular, $\ccangle{v}{v} = \ell_v$.
Analogously we define the sets $\coangle{u}{v}$, $\ocangle{u}{v}$ and $\ooangle{u}{v}$ which exclude one
or both of the half-lines.
We refer to an angle of the form $\ccangle{v}{-v}$ as \emph{half-plane}.
We write $v \clockwise u$ when
$u \in \ooangle{v}{-v}$, i.e., $u$ is oriented clock-wise with respect to $v$ (see Figure~\ref{fig:clock-order} for an illustration).
Note that $\clockwise$ defines a total order on pairwise non-colinear non-negative vectors.

By the \emph{cone} of a finite set of vectors $\{v_1, \ldots, v_k\}\subseteq\Z^2$ we mean the set of all non-negative rational 
linear combinations of these vectors:
\[
\cone(v_1, \ldots, v_k) := \setof{\Sigma_{j=1}^k\, a_j v_j \in \R^2
}{a_1, \ldots, a_k \in \Rnonneg}.
\] 
We call the cone of a single vector $\cone(v) = \ell_v$  \emph{trivial}, and the cone of zero vectors $\cone(\emptyset) = \{(0,0)\}$ \emph{degenerate}.
Two non-zero vectors $v_1$ and $v_2$ can be in four distinct relations:
{\sc (i)} they are colinear, {\sc (ii)} they are contralinear, {\sc (iii)} $v_1 \clockwise v_2$ and hence
$\cone(v_1, v_2) = \ccangle{v_1}{v_2}$, {\sc (iv)}
$v_2 \clockwise v_1$ and hence $\cone(v_1, v_2) = \ccangle{v_2}{v_1}$.
\begin{lemma}\label{lem:cone-dichotomy}
Every cone
either equals the whole plane $\R^2$, or is included in some half-plane.
\end{lemma}
\begin{proof}
Assume, w.l.o.g.~that the vectors $v_1, \ldots, v_k$ are non-zero and include no colinear pair.
Suppose there is a contralinear pair $v_i, v_j$  among $v_1, \ldots, v_k$. If all other vectors $v_h$ 
satisfy $v_i \clockwise v_h \clockwise v_j$
then $\cone(v_1, \ldots, v_k)$ is included in the half-plane $\ccangle{v_i}{v_j}$.
Otherwise $\cone(v_1, \ldots, v_k)$ is the whole plane.

Now suppose there is no contralinear pair among $v_1, \ldots, v_k$.
If some three $v_i, v_j, v_h$ of them satisfy
$
v_i \clockwise v_j \clockwise v_h \clockwise v_i
$
then $\cone(v_1, \ldots, v_k)$ includes the three angles 
$\ccangle{v_i}{v_j}$, $\ccangle{v_j}{v_h}$ and $\ccangle{v_h}{v_i}$, the union of which is the whole plane.
Otherwise, the relation $\clockwise$ is transitive and hence defines a (strict) total order on $\{v_1, \ldots, v_k\}$.
The minimal and maximal element $v_i$ and $v_j$ w.r.t.~the order satisfy $v_i \clockwise v_j$, and hence
$\cone(v_1, \ldots, v_k) = \ccangle{v_i}{v_j}$ is included in the half-plane $\ccangle{v_i}{-v_i}$.
\end{proof}

%
The \emph{sequential cone} of vectors $v_1, \ldots, v_k \in \Z^2$ imposes additional non-negativeness conditions,
namely for every $i$, the partial sum $a_1 v_1 + \ldots + a_i v_i$ must be non-negative (this is required later, when pumping cycles in a run whose effects are $v_1, \ldots, v_k$ in that order):
\[
\poscone(v_1, \ldots, v_k) := \setof{\Sigma_{j=1}^k\, a_j v_j \in \Rnonneg^2
}{a_1, \ldots, a_k \in \Rnonneg, \,
\prettyforall{i}{\,\Sigma_{j=1}^i\, a_j v_j
\in \Rnonneg^2}}.
\] 
Note that $v_1$ may be assumed w.l.o.g.~to be semi-positive, but
other vectors $v_i$ are not necessarily non-negative; and that every sequential cone is a subset of the non-negative orthant
$\Rnonneg^2$.
Importantly, contrarily to cones, the order of vectors $v_1, \ldots, v_k$ matters for sequential cones.
In fact, sequential cones are just convenient syntactic sugar for cones
of pairs of non-negative vectors:
\begin{lemma} \label{claim:seqcone}
For all vectors $v_1, \ldots, v_k$, the sequential cone 
$\poscone(v_1, \ldots, v_k)$, if not degenerate, equals $\cone(u, v)$, for two non-negative vectors $u$, $v$, 
and each of them  
either belongs to $\{v_1, \ldots, v_k\}$, or is horizontal, or vertical.
\end{lemma}

\subparagraph{Proof.}
We proceed by induction on $k$. For $k = 1$ we have $\poscone(v_1) = \ell_{v_1} = \cone(v_1, v_1)$.
Let $v_0$ and $h_0$ denote some fixed vertical and horizontal vector, respectively.
For the induction step we assume $\poscone(v_1, \ldots, v_{k-1}) = \cone(u, v)$ for non-negative vectors $u$, $v$;
and compute the value of $\poscone(v_1,  \ldots, v_k)$, separately in each of the following distinct cases
(assume w.l.o.g.~$u \clockwise v$):

\begin{wrapfigure}[13]{r}[0cm]{4cm}
  \centering
  \vspace{-0.5\baselineskip}
  \input{images/10-last-image.tikz}
\end{wrapfigure}

\vspace{-0.5\baselineskip}
\begin{align*}
&\poscone(v_1,  \ldots, v_k) = \\
&=\begin{cases}
\cone(v_k, v) 	& \text{if } v_k \in \coangle{v_0}{u} \\
\cone(u, v) 		& \text{if } v_k \in \ccangle{u}{v} \\
\cone(u, v_k) 	& \text{if } v_k \in \ocangle{v}{h_0} \\
\cone(u, h_0) 	& \text{if } v_k \in \ocangle{h_0}{-u} \\
\cone(v_0, h_0) 	& \text{if } v_k \in \ooangle{-u}{-v} \\
\cone(v_0, v) 	& \text{if } v_k \in \coangle{-v}{v_0}.
\end{cases}
\end{align*}
%
\qed

%% file: images/03-angle-order.tikz
\begin{tikzpicture}[scale=0.5]
    \draw[draw=cA!50,line width=2mm] (75+180:3.2cm) arc (75+180:360+45+180:3.2cm);
    \draw[draw=cA!50,line width=2mm,arrows={-Triangle[width=10pt, length=10pt]}] (70+180:3.15cm) -- (63+180:3.15cm);
    \node[point] at (0,0) {};
    \foreach \angle [count=\i] in {30,60,...,330} {
        \node[inner sep=0.05cm] (n\angle) at (45+180-\angle:{sin(\angle*3)*0.3+2.3}) {$u_{\i}$};
        \draw[vectorShadeA] (0,0) -- (n\angle);
    }
\end{tikzpicture}

%% file: images/10-last-image.tikz
\begin{tikzpicture}[scale=0.45]
    \drawGrid{6}
    \coordinate (S) at (0,0);
    \coordinate (p1) at (90:4cm);
    \coordinate (p2) at (60:4cm);
    \coordinate (p3) at (30:4cm);
    \coordinate (p4) at (0:4cm);
    \coordinate (q2) at (60+180:2.5cm);
    \coordinate (q3) at (30+180:2.5cm);

    \begin{pgfonlayer}{foreground}
    \draw[line width=1.0mm,draw=white,line cap=round] (S) -- (0,3.97);
    \draw[line width=1.0mm,draw=white,line cap=round] (S) -- (p2);
    \draw[line width=1.0mm,draw=white,line cap=round] (S) -- (p3);
    \draw[line width=1.0mm,draw=white,line cap=round] (S) -- (3.97,0);
    \draw[line width=1.0mm,draw=white,line cap=round] (S) -- (q2);
    \draw[line width=1.0mm,draw=white,line cap=round] (S) -- (q3);

    \draw[vector] (S) -- node[pos=0.8, auto]{$v_0$} (p1);
    \draw[vector] (S) -- node[pos=0.8, auto,swap]{$h_0$} (p4);
    \draw[vector] (S) -- node[pos=0.8, auto]{$u$} (p2);
    \draw[vector] (S) -- node[pos=0.8, auto,swap]{$v$} (p3);

    \draw[thick] (S) -- (q2);
    \draw[thick] (S) -- (q3);
    \end{pgfonlayer}

    \begin{pgfonlayer}{background}
        \begin{scope}
            \clip circle [radius=1.15cm];
            \fill[cA] (S) -- (p1) -- (p2); 
            \fill[cA] (S) -- (p3) -- (p4); 
            \fill[cA] (S) -- (q2) -- (q3); 
        \end{scope}
        \begin{scope}
            \clip circle [radius=1.25cm];
            \fill[cB] (S) -- (p2) -- (p3); 
            \fill[cB] (S) -- (p4) |- (q2); 
            \fill[cB] (S) -- (q3) |- (p1); 
        \end{scope}
    \end{pgfonlayer}
    \foreach \angle/\labelll in {75/I,45/II,15/III,300/IV,225/V,135/VI} {
        \node at (\angle:1.8cm) {\labelll};
    }
\end{tikzpicture}

%% file: thin-thick.tex
\newpage
\section{Thin-Thick Dichotomy} \label{sec:thin-thick}

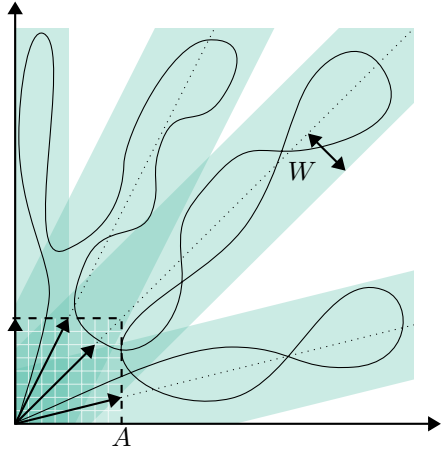
\begin{wrapfigure}[18]{r}[0cm]{7cm}  
  \centering
  \vspace{-0.5\baselineskip}
  \vspace{-\baselineskip}
  \input{images/01-thin-run.tikz}
  \caption{Thin run within belts $B_{v,W}$ .} 
  \label{fig:thin-run}
\end{wrapfigure}

The main result of this section (cf.~\cref{lem:run-2d} below) classifies $\zero$-runs in a \nvass{2}
into \emph{thin} and \emph{thick} ones.
Throughout this section we consider an arbitrary fixed \twovass $V = (Q, T)$.
\newline Let $n = \size{Q}$ and $M = \norm{V}$.

\subparagraph*{Thin runs.}

The \emph{belt} of \emph{direction} $v \in \N^2$ and \emph{width} $W$ is the set 
\[
\belt{v}{W} = \{u \in \N^2 \mid  \dist(u, \ell_v) \leq W \},
\]
where $\dist(u, \ell_v)$ denotes the Euclidean distance between the point $u$ and the half-line $\ell_v$.
For $A\in\N$, we call $\belt{v}{W}$ an \emph{$A$-belt} if $\norm{v} \leq A$ and $W \leq A$.
We say that a run $\rho$ of $V$ is \emph{$A$-thin} if for every configuration $c$ in $\rho$ there exists an $A$-belt $B$ such that $c \in Q \times B$.

\subparagraph*{Thick runs.}

Let $A\in\N$.
Four cycles $\pi_1, \pi_2, \pi_3, \pi_4 \in T^*$ are \emph{$A$-sequentially enabled} in a run $\rho$ if 
their lengths are at most $A$, and 
the run $\rho$ factors into $\rho = \rho_1 \, \rho_2 \, \rho_3 \, \rho_4 \, \rho_5$ so that
(denote by $v_1, v_2, v_3, v_4$ the effects of $\pi_1, \pi_2, \pi_3, \pi_4$, respectively): 
\begin{itemize}
\item The effect $v_1$ is semi-positive, the cycle $\pi_1$ is enabled in $c_1 := \trg(\rho_1)$, and both coordinates 
are bounded by $A$ along $\rho_1$.
\item If $v_1$ is positive then $\pi_2$  is $\emptyset$-enabled in $c_2 := \trg(\rho_2)$.
Otherwise (let $j$ be the coordinate s.t.~$v_1[j] = 0$)
$\pi_2$ is $\{j\}$-enabled in $c_2 := \trg(\rho_2)$, and $j$th coordinate is bounded by $A$ along $\rho_2$.
\item The cycle $\pi_i$ is $\emptyset$-enabled in $c_i := \trg(\rho_i)$, for $i = 3,4$.
\end{itemize}
%
We also say that the four vectors $v_1, v_2, v_3, v_4$ are $A$-sequentially enabled in $\rho$, quantifying 
the cycles existentially.
A $\zero$-run $\tau$ is called \emph{A-thick} if it partitions into $\tau = \rho \, \rho'$ so that
\begin{enumerate}
\item some vectors $v_1, v_2, v_3, v_4$ 
are $A$-sequentially enabled in $\rho$,
\item some vectors $v'_1, v'_2, v'_3, v'_4$ 
are $A$-sequentially enabled in $\rev(\rho')$,
\item $\poscone(v_1, v_2, v_3, v_4) \cap \poscone(v'_1, v'_2, v'_3, v'_4)$ is non-trivial.
\end{enumerate}
 
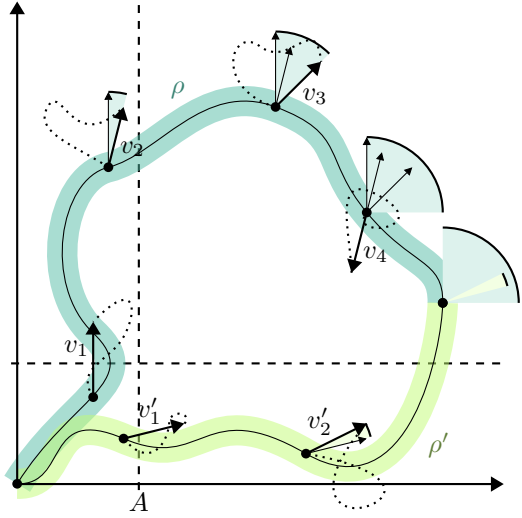
\begin{wrapfigure}[23]{r}[0cm]{7.7cm}  
  \centering
  \input{images/02-thick-run.tikz}
  \caption{Thick run. Blue angles denote sequential cones $\poscone(v_1, v_2)$, $\poscone(v_1, v_2, v_3)$ and 
  $\poscone(v_1, v_2, v_3, v_4)$, respectively, and green angle denotes $\poscone(v'_1, v'_2)$.}
  \label{fig:thick-run}
\end{wrapfigure}

\newpage
\noindent
Figure~\ref{fig:thick-run} illustrates the geometric ideas underlying these three conditions for $A$-thick runs.
Concerning condition 1, a cycle $\pi_1$ depicted by a dotted line,
with vertical effect $v_1$,   
can be used to increase the second (vertical) coordinate arbitrarily, which justifies the relaxed requirement that
a cycle $\pi_2$ with effect $v_2$ is only $\{1\}$-enabled.
Note that the norm of the configuration enabling $\pi_1$, as well as the first coordinate of the configuration enabling $\pi_2$,
are bounded by $A$.
Concerning condition 2, a cycle $\pi'_1$ with positive effect $v'_1$ can be used to increase both coordinates arbitrarily; therefore
a cycle $\pi'_2$ with effect $v'_2$ is only required to be $\emptyset$-enabled, and no coordinate of the configuration enabling $\pi'_2$ is
required to be bounded by $A$.
In the illustrated example, vectors $v'_3$ and $v'_4$ are not needed; 
formally, one can assume $v'_2 = v'_3 = v'_4$ and $\rho'_3 = \rho'_4 = \emptyseq$.
Condition 3 ensures that the cycles $\pi_1, \ldots, \pi_4$ and
$\pi_1', \ldots, \pi_4'$ can be pumped such that the pumped versions of $\rho$ and $\rho'$
are still connected. In the illustrated example, observe that 
$\poscone(v_1, v_2) \cap \poscone(v'_1) = \emptyset$.
Intuitively, both coordinates in the target of $\rho$ can be increased arbitrarily using $v_1$ and $v_2$,
and similarly both coordinates of the target of $\rev(\rho')$ can be increased arbitrarily using $v'_1$,
but `directions of increase' are non-crossing.
Adding $v_3$ and $v'_2$ is not sufficient, as still
$\poscone(v_1, v_2, v_3) \cap \poscone(v'_1, v'_2) = \emptyset$.
When vector $v_4$ is adjoined, condition 3 holds as $\poscone(v_1, v_2, v_3, v_4) = \Rnonneg^2$.
Finally, the four vectors are really needed here, e.g., vector $v_3$ can not be omitted as
$\poscone(v_1, v_2, v_4) = \poscone(v_1, v_2)$.

Here is the main result of this section:  

\begin{theorem}[Thin-Thick Dichotomy]\label{lem:run-2d}
There is a polynomial $p$ such that 
every $\zero$-run in a \twovass $V$ is either $p(nM)^n$-thin or $p(nM)^n$-thick.
\end{theorem}
For the proof of \cref{lem:run-2d} we need the following core fact (shown in the appendix):   
\begin{lemma}[Non-negative Cycle Lemma] \label{lem:first-cycle}
There is a polynomial $P$ such that every run $\rho$ in $V$ from a $\zero$-configuration to a
target configuration of norm larger than $P(n M)^n$, contains a configuration enabling a  
semi-positive cycle of length at most $P(n M)$.
\end{lemma}

%
%

\subparagraph*{Proof of \cref{lem:run-2d}.}
%
%

Let $P$ be the polynomial from \cref{lem:first-cycle}.  
The polynomial $p$ required in \cref{lem:run-2d} can be chosen arbitrarily as long as 
%
$p(x) \geq \sqrt 2 \cdot \big(P(x) + (x+1)^3\big) \cdot x$.
%
for all $x$; note that the following inequality follows:
\begin{align} \label{eq:disj2}
p(n M)^n \geq \sqrt 2 \cdot \big(\left(P(n M)\right)^n + (n M+1)^3\big) \cdot nM.
\end{align}
%
%

In the sequel we deliberately confuse configurations $c = (q, v)$ with \emph{their vectors $v$}: whenever convenient,
we use $c$ to denote the vector $v$, hoping that this does not lead to any confusion.

Let $\tau$ be a $\zero$-run of $V$ which is not $p(n M)^n$-thin, i.e., $\tau$ contains therefore
a configuration $t$ which lies outside of all the $p(n M)^n$-belts. 
We need to demonstrate points 1--3 in the definition of thick run.
To this aim we split $\tau$ into $\tau = \rho \, \rho'$ where $\trg(\rho) = t = \src(\rho')$, and are going to prove the following two claims (a) and (a'). Let $D := P(nM)^n + (n M + 1)^3$. 
For $x, y \in \R^2$, let $\dist(x, y)$ denote their Euclidean distance.

\begin{enumerate}
\item[(a)] Some vectors $v_1, v_2, v_3, v_4$ are $P(nM)^n$-sequentially enabled in $\rho$, 
and the sequential cone $\poscone(v_1, v_2, v_3, v_4)$ contains a point $u \in \Rnonneg^2$ with $\norm{u-t} \leq D$.
\item[(a')] Some vectors $v'_1, v'_2, v'_3, v'_4$ are $P(nM)^n$-sequentially enabled in $\rev(\rho')$, 
and the sequential cone $\poscone(v'_1, v'_2, v'_3, v'_4)$ contains a point $u \in \Rnonneg^2$ with $\norm{u-t} \leq D$.
\end{enumerate}

\noindent
In simple words, instead of proving point 3, we prove that both sequential cones contain a point $v$ which is sufficiently
close to $t$.

\begin{claim} \label{claim:a}
The conditions (a) and (a') guarantee that $\tau$ is thick.
\end{claim}
Indeed, points 1--2 in the definition of thick run are immediate as $P(n M) \leq p(n M)$.
For point 3, 
observe that the inequality~\eqref{eq:disj2} implies $p(nM)^n \geq \sqrt 2 \cdot D$, which guarantees that
the circle $\setof{u\in\Rnonneg^2}{\dist(u, t) \leq \sqrt 2 \cdot D}$ does not touch any half-line $\ell_w$ induced by 
a non-negative vector $w$ with $\norm{w} \leq p(nM)^n$.
In consequence, neither does the square $X := \setof{u\in\Rnonneg^2}{\norm{u-t} \leq D}$
inscribed in the circle, and hence $X$ 
lies between two consecutive half-lines $\ell_w$ induced by a non-negative vector $w$ with $\norm{w} \leq p(nM)^n$.
Hence, as $\poscone(v_1, v_2, v_3, v_4)$ contains some point of $X$, by \cref{claim:seqcone} it includes the whole $X$,
and likewise $\poscone(v'_1, v'_2, v'_3, v'_4)$. In consequence, the whole $X$ is included in 
$\poscone(v_1, v_2, v_3, v_4) \cap \poscone(v'_1, v'_2, v'_3, v'_4)$ which entails point 3.
\Cref{claim:a} is thus proved.

As condition (a') is fully symmetric to (a), we focus exclusively on proving condition (a), i.e., on
constructing sequentially enabled vectors $v_1, v_2, v_3, v_4$.

Vector $t$ lies outside of $p(nM)^n$-belts, hence outside of all the $P(nM)^n$-belts, therefore
its norm $\norm{t} > P(n M)^n$.
Relying on \cref{lem:first-cycle},
let $c_1$ be the first configuration in the run $\rho$ which enables a semi-positive cycle $\pi_1$
of length bounded by $P(n M)$, and let $v_1 = \eff(\pi_1)$.
We start with the following obvious claim (let $v_0$ be some vertical vector, e.g. $v_0 = (0,1)$):
\begin{claim}\label{claim:0}
$\poscone(v_0)$ contains a point $u\in \Rnonneg^2$ such that $\norm{u - c_{1}} \leq P(nM)^n + n M$.
\end{claim}
Indeed, due to \cref{lem:first-cycle} we may assume $\norm{c_1} \leq P(nM)^n + M$ and hence
$u = \zero$ does the job.

Recall that the relation $\clockwise$ defines a total order on
pairwise non-colinear non-negative vectors.
\begin{claim} \label{claim:wlog1}
We can assume w.l.o.g.~that $v_1 \clockwise t$.
\end{claim}
Indeed, if $v_1$ and $t$ were colinear then $t\in\cone(v_1)$ and hence condition (a) would hold.

Split $\rho$ into the prefix ending in $c_1$ and the remaining suffix: 
$\rho = \rho_1 \, \sigma$, where $\trg(\rho_1) = c_1 = \src(\sigma)$.
As the next step we will identify a configuration $c_2$ in $\sigma$ which 
satisfies \cref{claim:1} (which will serve later as the basis of induction)
and enables a cycle $\pi_2$ with effect $v_2$ (as stated in \cref{claim:cycle}).

\begin{claim}\label{claim:1}
$\poscone(v_0, v_1)$ contains a point $u\in\Rnonneg^2$ such that $\norm{u - c_{2}} \leq P(nM)^n + 2 n M$.
\end{claim}

\begin{wrapfigure}[17]{r}[0cm]{4cm}
  \centering
  \vspace{-0.5\baselineskip}
  \input{images/09-early-loop.tikz}
\end{wrapfigure}

\noindent
The proof of \cref{claim:1} depends on whether $v_1$ is positive. 
If $v_1$ is so, we simply duplicate the first cycle: $c_2 := c_1$ and $\pi_2 := \pi_1$, and use \cref{claim:0}.
Otherwise $v_1$ is vertical due to \cref{claim:wlog1}.  
If $t[1] \leq W = P(n M)^n + (n+1) M$ then condition (a) holds immediately as 
$\poscone(v_1) = \ell_{v_1}$ contains a point $u\in\Rnonneg^2$
with $\norm{u-t} \leq P(n M)^n + (n+1) M \leq D$.
Therefore suppose $t[1] > P(n M)^n + (n+1) M$, and define the sequence $d_1, \ldots, d_m$ of configurations
as follows: let $d_1 := c_1$, and let $d_{i+1}$ be the first configuration in $\sigma$ with $d_{i+1}[1] > d_i[1]$.
Recall that $d_1[1] \leq P(n M)^n + M$, and observe that $d_{i+1}[1] \leq d_i[1]+M$.
Thus by the pigeonhole principle $m > n$ and hence for some $i<j\leq n+1$ the configurations 
$d_i$ and $d_j$ must have the same control state.
The infix $\sigma_{ij}$ of the path $\sigma$ from $d_i$ to $d_j$ is thus a cycle, 
enabled in $d_i$, whose effect is positive on the first (horizontal) coordinate.
Let $c_2 := d_i$.
As $c_2[1] \leq P(n M)^n + (n+1) M$, $\poscone(v_0, v_1) = \ell_{v_0}$ contains necessarily 
a point $u\in\Rnonneg^2$ such that $\norm{u - c_{2}} \leq P(nM)^n + (n+1) M$, which proves
\cref{claim:1}.

\begin{claim} \label{claim:cycle}
The configuration $c_2$ $\{1\}$-enables a cycle $\pi_2$ of length bounded by $p(n M)^n$, 
such that the first coordinate of $\eff(\pi_2)$ is positive.
\end{claim}
Recalling the proof of the previous claim, 
observe that the first (horizontal) coordinate in the infix $\sigma_{ij}$ is bounded
by $P(n M)^n + (n+1) M$, and think of the second (vertical) coordinate as irrelevant.
Let $\pi_2$ be the path inducing $\sigma_{ij}$. For bounding the length of $\pi_2$,
as long as $\pi_2$ contains a cycle $\alpha$ with vertical effect $(0, w)$, remove $\alpha$ from 
$\pi_2$.
This process ends yielding a cycle $\pi_2$ of length at most $(P(n M)^n + (n+1) M) \cdot n$, and hence at most $p(n M)^n$
(by the inequality~\eqref{eq:disj2}),
which is $\{1\}$-enabled in $c_2$, but not necessarily enabled.
Let $v_2 := \eff(\pi_2)$.

\begin{claim}
We can assume w.l.o.g.~that $v_2 \clockwise t$.
\end{claim}
Indeed, if $v_1 = v_2$ then \cref{claim:wlog1} does the job;
otherwise $v_1$ is vertical and then $t \clockwise v_2$ (or $t$ colinear with $v_2$) would imply 
$t \in \poscone(v_1, v_2)$, hence condition (a) would hold again.

\vspace{2mm}

Split $\sigma$ further into the prefix ending in $c_2$ and the remaining suffix: 
$\sigma = \rho_2 \sigma'$, where $\trg(\rho_2) = c_2 = \src(\sigma')$.
%
If $\sigma'$ contains a configuration
which $\emptyset$-enables a simple cycle whose effect $w$ belongs to $\coangle{t}{-v_2}$ 
then $t \in \poscone(v_2, w)$ and hence condition (a) holds. 
We aim at achieving this objective incrementally.

\begin{wrapfigure}[10]{r}[0cm]{5.5cm}
  \centering
  \vspace{-\baselineskip}
  \input{images/04-increasing-angle.tikz}
  \label{fig:increasing-angle}
\end{wrapfigure}

For $i \geq 2$, let $c_{i+1}$ be the first configuration in $\sigma'$
after $c_i$ that $\emptyset$-enables a simple cycle $\pi_{i+1}$ with effect $v_{i+1} \in \ooangle{v_i}{-v_i}$.
As discussed above, if $v_{i+1} \in \coangle{t}{-v_i}$ for some $i$ then $t \in \poscone(v_i, v_{i+1})$ 
and hence condition (a) holds.  
Assume therefore that the sequence $v_1, \ldots, v_m$ so defined satisfies $v_{i+1} \in \ooangle{v_i}{t}$ 
for all $i \geq 2$. Let $c_{m+1} := t$.
As vectors $v_3, \ldots, v_m$ are pairwise different, semi-positive and, 
being effects of simple cycles, have norms at most $nM$,
we know that $m \leq (nM+1)^2 + 1$.

\begin{claim}\label{claim:final}
For every $i = 1, \ldots, m$, $\poscone(v_0, v_i)$ contains a point $u\in\Rnonneg^2$ such
that $\norm{u - c_{i+1}} \leq P(nM)^n + (i+1) n M$.
\end{claim}

\begin{proof}
By induction on $i$. The induction base is exactly \cref{claim:1}.
For the induction step, we are going to show that
$\poscone(v_0, v_i)$ contains a vector $u$ such that $\norm{u - c_{i+1}} \leq P(nM)^n + (i+1) n M$.
Decompose the infix of $\sigma'$ which starts in $c_i$ and ends in $c_{i+1}$ into simple cycles, 
plus the remaining path $\bar \rho$ of length at most $n$.
The norm of the effect $\bar v$ of $\bar \rho$ is hence bounded by $n M$, and we have
$$c_{i+1} \ =  \ c_i + s + \bar v,$$
where $s$ is the sum of effects of all the simple cycles. 
By the definition of $v_{i+1}$, the effects of all the simple cycles belong to the half-plane
$\ccangle{-v_i}{v_i}$, and hence
there belongs $s$. 
By induction assumption there is $u' \in \poscone(v_0, v_{i-1})$ such that $\norm{u' - c_i} \leq P(nM)^n + inM$. 
As $v_{i-1} \clockwise v_i$, we also have $u' \in \poscone(v_0, v_i)$.
Consider the point   
$$u \ := \ u' + s$$
which necessarily belongs to the half-plane $\ccangle{-v_i}{v_i}$ but 
not necessarily to $\poscone(v_0, v_i) = \ccangle{-v_i}{v_i} \cap \Rnonneg^2$.
Ignoring this issue, by routine calculations we get
\[
\norm{u - c_{i+1}}  \ = \ \norm{u' + s - c_i - s - \bar v} \ \leq \ \norm{u' - c_i} + \norm{\bar v} \leq \norm{u' - c_i} + n M
\leq P(n M)^n + (i+1) n M 
\]
as required for the induction step.
Finally, if $u \notin\Rnonneg^2$, translate $u$ towards $c_{i+1}$ until it enters the non-negative orthant $\Rnonneg^2$;
clearly, the translation can only decrease the value of $\norm{u - c_{i+1}}$.
\end{proof}
Applying the claim to $i = m$, and knowing that $m \leq (nM+1)^2 + 1$, 
we get some point $u \in \poscone(v_0, v_m)$ such that 
$\norm{u - t} \leq P(nM)^n + ((nM+1)^2 + 1) \cdot n M \leq P(nM)^n + (nM+1)^3$.
Furthermore, relying on the assumptions that $t$ lies outside of all $p(nM)^n$-belts and 
that $v_1 \clockwise t$ we prove, similarly as in the proof of \cref{claim:a}, that $v_1 \clockwise u$ and hence
the point $u$ belongs also to $\poscone(v_1, v_m)$.
This completes the proof of \cref{lem:run-2d}.
\qed

%% file: images/01-thin-run.tikz
\begin{tikzpicture}[scale=0.35]
    \newcommand{\gridSize}{16}
    \drawFading{\gridSize}
    \drawGrid{\gridSize}
    \draw[step=0.5cm,white,thin,draw opacity=0.7] (0,0) grid (4-0.01,4-0.01);
    \foreach \vx/\vy in {0/4,2/4,3/3,4/1} {
        \begin{scope}
            \clip (0,0) rectangle (\gridSize,\gridSize);
            \drawBelt{\vx}{\vy}{2cm}{\gridSize}
        \end{scope}
        \draw[style=vector] (0,0) -- (\vx,\vy);
    }

    \draw[black,thick,dashed] (-0.1,4) -- (4,4) -- (4,-0.1);
    \node at (4,-0.5) {$A$};

    \begin{scope}[y=1.97, x=1.97, yshift=0, inner sep=0pt, outer sep=0pt]
        \path[runLineBlack] (1.0000,1.0000) .. controls (6.3333,17.1667) and
        (17.3000,53.5000) .. (18.5000,69.5000) .. controls (20.0000,89.5000) and
        (2.0000,114.5000) .. (2.0000,157.0000) .. controls (2.0000,184.0000) and
        (7.5000,217.5000) .. (15.5000,213.5000) .. controls (32.5000,205.0000) and
        (1.5000,84.0000) .. (27.0000,95.0000) .. controls (52.5000,106.0000) and
        (58.5000,126.0000) .. (59.0000,142.0000) .. controls (59.5000,158.0000) and
        (84.5000,208.5000) .. (104.5000,210.0000) .. controls (124.5000,211.5000) and
        (124.0000,195.5000) .. (112.0000,178.0000) .. controls (100.0000,160.5000) and
        (88.0000,175.5000) .. (80.0000,155.5000) .. controls (72.0000,135.5000) and
        (80.0000,131.5000) .. (74.0000,115.5000) .. controls (68.0000,99.5000) and
        (47.5000,111.0000) .. (38.0000,88.5000) .. controls (28.5000,66.0000) and
        (29.0000,53.0000) .. (52.0000,42.0000) .. controls (68.9556,33.8908) and
        (92.8600,55.5000) .. (90.5000,85.0000) .. controls (88.5000,110.0000) and
        (89.0000,113.5000) .. (108.5000,135.5000) .. controls (130.1210,159.8930) and
        (134.5000,143.1840) .. (168.0000,152.0000) .. controls (193.5000,158.7100) and
        (211.0000,171.0000) .. (201.5000,190.5000) .. controls (190.7610,212.5440) and
        (171.0000,204.0000) .. (160.0000,187.5000) .. controls (149.0000,171.0000) and
        (145.0000,138.0000) .. (126.5000,110.5000) .. controls (108.4020,83.5979) and
        (68.3155,68.5824) .. (59.0000,47.5000) .. controls (49.5000,26.0000) and
        (87.0000,6.5000) .. (117.5000,14.0000) .. controls (151.6140,22.3888) and
        (170.5000,86.9667) .. (201.5000,72.5000) .. controls (216.5000,65.5000) and
        (214.5000,33.1379) .. (183.5000,31.0000) .. controls (152.5000,28.8620) and
        (150.6280,42.0000) .. (117.5000,42.0000) .. controls (83.5000,42.0000) and
        (34.1667,14.1667) .. (1.0000,1.0000) -- cycle;
    \end{scope}

    \coordinate (width1) at (11,11);
    \coordinate (width2) at ($(width1)!2cm!90:(0,0)$);
    \draw[basicArrowBoth] (width1) -- node[auto,midway,swap] {$W$} (width2);
\end{tikzpicture}

%% file: images/02-thick-run.tikz
\begin{tikzpicture}[scale=0.40]
    \newcommand{\gridSize}{16}
    \clip (-0.5,-1.2) rectangle (\gridSize+0.5,\gridSize);
    \drawGrid{\gridSize}

    \draw[black,thick,dashed] (-0.2,4) -- (\gridSize,4);
    \draw[black,thick,dashed] (4,-0.2) -- (4,\gridSize);
    \node at (4,-0.6) {$A$};

    \begin{scope}[yshift=16cm]
    \begin{scope}[xscale=0.05,yscale=-0.05]
        \begin{scope}
            \draw[runOutlineA]    (280,200) .. controls (280.5,171) and (259.5,176) .. (230,140) .. controls (200.5,104) and (217.5,88) .. (170,70) .. controls (122.5,52) and (92.5,102) .. node[auto,pos=0.4,swap,inner sep=0]{\color{cA!70!black}$\rho$} (60,110) .. controls (27.5,118) and (16.5,188) .. (50,220) .. controls (83.5,252) and (34.5,264) .. (0,320) ;
            \draw[runOutlineB]    (0,320) .. controls (34.5,321) and (16.5,267) .. (70,290) .. controls (123.5,313) and (121.5,260) .. (190,300) .. controls (258.5,340) and (279.5,229) .. node[auto,pos=0.5,swap,inner sep=0]{\color{cB!50!black}$\rho'$} (280,200);
        \end{scope}

        \begin{scope}[even odd rule]
            \clip (90,-10) -- (60,0) -- (60,110) -- cycle; 
            \draw[sequentialConeAngleA] (60,110) circle (50);
        \end{scope}
        \begin{scope}[even odd rule]
            \clip (280,-40) -- (170,-40) -- (170,70) -- cycle;
            \draw[sequentialConeAngleA] (170,70) circle (50);
        \end{scope}
        \begin{scope}[even odd rule]
            \clip (320,140) -- (320,0) -- (230,-10) -- (230,140) -- cycle;
            \draw[sequentialConeAngleA] (230,140) circle (50);
        \end{scope}
        \begin{scope}[even odd rule]
            \clip (380,250) -- (370,210) -- (190,300) -- cycle;
            \draw[sequentialConeAngleB] (190,300) circle (44);
        \end{scope}

        \begin{scope}[even odd rule,xshift=50cm,yshift=60cm]
            \clip (320,140) -- (320,0) -- (230,-10) -- (230,140) -- cycle;
            \draw[sequentialConeAngleA] (230,140) circle (50);
        \end{scope}
        \begin{scope}[even odd rule,xshift=90cm,yshift=-100cm]
            \clip (380,250) -- (370,210) -- (190,300) -- cycle;
            \draw[sequentialConeAngleB] (190,300) circle (44);
        \end{scope}

        \draw[vector]    (190,300) -- node[vectorLabel,pos=0.5] {$v_2'$} (230,280) ;
        \draw[vector,yshift=42.5cm]    
                         (50,220)  -- node[vectorLabel,pos=0.65] {$v_1$}  (50,170) ;
        \draw[vector]    (60,110)  -- node[vectorLabel,pos=0.5,swap] {$v_2$}  (70,70) ;
        \draw[vector]    (230,140) -- node[vectorLabel,pos=0.5] {$v_4$}  (220,180) ;
        \draw[vector]    (170,70)  -- node[vectorLabel,pos=0.5,swap] {$v_3$}  (200,40) ;
        \draw[vector]    (70,290)  -- node[vectorLabel,pos=0.7] {$v_1'$} (110,280) ;

        \node[point] at (280,200) {};
        
        \draw[vectorShadeA] (60,110) -- (60,60) ;
        \draw[vectorShadeA] (170,70) -- (180,30) ;
        \draw[vectorShadeA] (170,70) -- (170,20) ;
        \draw[vectorShadeA] (230,140) -- (260,110) ;
        \draw[vectorShadeA] (230,140) -- (240,100) ;
        \draw[vectorShadeA] (230,140) -- (230,90) ;

        \draw[vectorShadeB] (190,300) -- (230,290) ;

        
        \draw [dotted,thick,color=black,draw opacity=1 ]   (230,280) .. controls (191.5,324) and (211.5,346-5) .. (230,340-5) .. controls (248.5,334-5) and (251.5,311) .. (190,300) ;
        \draw [dotted,thick,color=black,draw opacity=1 ]   (110,280) .. controls (116.5,274) and (100.5,269) .. (100,280) .. controls (99.5,291) and (88.5,312) .. (70,290) ;
        \draw [dotted,thick,color=black,draw opacity=1,yshift=42.5cm ]   (50,220) .. controls (37.5,200) and (60.5,197) .. (70,180) .. controls (79.5,163) and (78.5,140) .. (50,170) ;
        \draw [dotted,thick,color=black,draw opacity=1 ]   (60,110) .. controls (37.5,95) and (10.5,87) .. (20,70) .. controls (29.5,53) and (36.5,103) .. (70,70) ;
        \draw [dotted,thick,color=black,draw opacity=1 ]   (170,70) .. controls (128.5,58) and (140.5,15) .. (160,30) .. controls (179.5,45) and (196.5,6) .. (200,40) ;
        \draw [dotted,thick,color=black,draw opacity=1 ]   (220,180) .. controls (219.5,154) and (205.5,109) .. (240,130) .. controls (274.5,151) and (216.5,158) .. (230,140) ;
        
        \draw[runLineBlack]    (0,320) node[point]{} .. controls (34.5,321) and (16.5,267) .. (70,290) node[point]{} .. controls (123.5,313) and (121.5,260) .. (190,300) node[point]{} .. controls (258.5,340) and (279.5,229) .. (280,200);
        \draw[runLineBlack]    (280,200) node[point]{} .. controls (280.5,171) and (259.5,176) .. (230,140) node[point]{} .. controls (200.5,104) and (217.5,88) .. (170,70) node[point]{} .. controls (122.5,52) and (92.5,102) .. (60,110) node[point]{} .. controls (27.5,118) and (16.5,188) .. (50,220) node[point,yshift=-0.85cm]{} .. controls (83.5,252) and (34.5,264) .. (0,320) ;

    \end{scope}
    \end{scope}
\end{tikzpicture}

%% file: images/09-early-loop.tikz
\begin{tikzpicture}[scale=0.40]
    \drawGridB{6}{14}
    \node at (5.2,-0.6) {$W$};
    \draw[black,thick,dashed] (5.2,-0.2) -- (5.2,14);

    \begin{scope}[xshift=-5mm]
    \begin{scope}[xscale=0.05,yscale=0.05]

    \draw[runLineBlack]
        (10,0) node[](a){} 
        .. controls (25.0,2) and (39.5,20) .. 
        (40,40) node[point,label=right:{$c_1=d_1$}](b){} 
        .. controls (40.5,60) and (31.5,65) .. 
        (30,80) node[](c){} 
        .. controls (28.5,95) and (41.5,103) .. 
        (50,110) node[point,label=right:{$d_2$}](d){} 
        .. controls (58.5,117) and (60.5,137) .. 
        (60,140) node[point,label=right:{$d_3$}](e){} 
        .. controls (59.5,143) and (57.5,158) .. 
        (50,170) coordinate(f) ;
    \draw[runLineBlack,decoration={zigzag,segment length=1mm, amplitude=0.25mm},decorate,very thick]
        (f) 
        .. controls (42.5,182) and (43.5,201) .. 
        node[midway,auto,xshift=-1mm] () {$\tau$}
        (50,210) coordinate (g) ;
    \draw[runLineBlack]
        (g)
        .. controls (56.5,219) and (62.5,224) .. 
        (70,230) node[point,label=right:{$d_4$}](h){}
        .. controls (77.5,236) and (84.5,243) .. 
        node[midway,auto] {$\pi_2$}
        (88,250) node[point,label=right:{$d_5$\ }](i){} 
        .. controls (95.5,257) and (108.5,268) .. 
        (120,280) node[](j){} ;
    \node[point] at (f) {};
    \node[point] at (g) {};
    
    \begin{pgfonlayer}{background}
    \draw[runOutlineA,line cap=round]
    (50,110)
    .. controls (58.5,117) and (60.5,137) .. 
    (60,140)
    .. controls (59.5,143) and (57.5,158) .. 
    (50,170)
    .. controls (42.5,182) and (43.5,201) .. 
    (50,210)
    .. controls (56.5,219) and (62.5,224) .. 
    (70,230)
    .. controls (77.5,236) and (84.5,243) .. 
    (90,250)
    ;
    \draw[runOutlineA,draw=white,line width=3mm,line cap=round,draw opacity=1]
    (50,170)
    .. controls (42.5,182) and (43.5,201) .. 
    (50,210);
    \end{pgfonlayer}

    \end{scope}
    \end{scope}

    \foreach \idx in {b,d,e,h,i} {
        \path let \p1 = (\idx) in coordinate (projection) at (\x1,14);
        \draw[dotted] (\idx) -- (projection);
    }
\end{tikzpicture}

%% file: images/04-increasing-angle.tikz
\begin{tikzpicture}[scale=0.25]
    \newcommand{\gridSize}{16}
    \drawGridB{\gridSize}{16}

    \begin{scope}[yshift=16cm]
    \begin{scope}[xscale=0.05,yscale=-0.05]
        \draw[runOutlineA] (280,160) node (t) [point] {} .. controls (217.5,130) and (176.5,129) .. node (p3) [pos=0.45,point] {} (110,180) node (p2) [point] {} .. controls (43.5,231) and (56.5,264) .. node (p1) [pos=0.7,point] {} (0,320) ;
        \draw[runLineBlack] (280,160) node (t) [point] {} .. controls (217.5,130) and (176.5,129) .. node (p3) [pos=0.45,point] {} (110,180) node (p2) [point] {} .. controls (43.5,231) and (56.5,264) .. node (p1) [pos=0.7,point] {} (0,320) ;
    \end{scope}
    \end{scope}
    \draw[vector] (p1.center) -- node[vectorLabel,pos=1.1,swap]{$v_2$} ++(80:4cm);
    \draw[vectorShadeA] (p2.center) -- ++(80:4cm);
    \draw[vector] (p2.center) -- node[vectorLabel,pos=1.05,swap]{$v_3$} ++(55:4cm);
    \draw[vectorShadeA] (p3.center) -- ++(80:4cm);
    \draw[vectorShadeA] (p3.center) -- ++(55:4cm);
    \draw[vector] (p3.center) -- node[vectorLabel,pos=1,swap]{$v_4$} ++(30:4cm);
    \path (t) ++(0.5,-0.5) node {$t$};
    \path (p1) ++(1.0,-0.9) node {$c_2$};
    \path (p2) ++(0.8,-1.0) node {$c_3$};
    \path (p3) ++(0.0,-1.2) node {$c_4$};
\end{tikzpicture}

%% file: proof-reach.tex
\section{Dichotomy in Action} \label{sec:app}

This section illustrates applicability of \cref{lem:run-2d}.
As before, we use symbols $n$ and $M$ for the number of control states, and the norm of a \twovass, respectively.
As the first corollary we provide a pumping lemma for \twovass: in case of thin runs apply, essentially,
pumping schemes of \onevass, and in case of thick runs use the cycles enabled along a run.
As another application, we derive an alternative proof of the exponential run property for \twovass.
\begin{theorem}[Pumping] \label{thm:pumping}
There is a polynomial $p$ such that every $\zero$-run $\tau$ in a \twovass of length greater that $p(n M)^n$ 
factors into
$
\tau \ = \ \tau_0 \, \tau_1 \, \ldots \, \tau_k
$
($k \geq 1$),
so that for some non-empty cycles $\alpha_1, \ldots, \alpha_k$ of length at most $p(n M)^n$, the path
$
\tau_0 \, \alpha_1^i \, \tau_1 \, \alpha_2^i \, \ldots, \, \alpha_k^i \, \tau_k
$
is a $\zero$-run for every $i\in\N$.
Furthermore, the lengths of $\tau_0$ and $\tau_k$ are also bounded by $p(n M)^n$.
\end{theorem}
\begin{theorem}[Exponential run] \label{thm:short-run}
There is a polynomial $p$ such that for every $\zero$-run $\tau$ in a \twovass, there is a $\zero$-run
of length bounded by $p(n M)^n$ with the same source and target as $\tau$.
\end{theorem}
We fix from now on a \twovass $V = (Q, T)$ and the polynomial $p$ of \cref{lem:run-2d}. Let $A = p(n M)^n$.
Both proofs proceed separately for thin and thick runs $\tau$. The former (fairly standard) 
case is moved to the appendix, so assume below $\tau$ to be $A$-thick. 
The polynomials required in \cref{thm:pumping,thm:short-run} can be read out from the constructions.

We rely on the standard tool, cf.~Prop.~2 in~\cite{taming}
(the norm of a system of inequalities is the largest absolute value of its coefficient, 
and likewise we define the norm of a solution):
%
%
%
\input{proof-reach-thick}

%% file: proof-reach-thick.tex
%
%
\begin{lemma} \label{lem:sol}
Let $\U$ be a system of $d$ linear inequalities  of norm $M$ with $k$ variables.
Then the smallest norm of a non-negative-integer solution of $\U$ is in ${\cal O}(k \cdot M)^d$.
\end{lemma}

Consider a split $\tau = \rho \rho'$, where 
$\rho = \rho_1 \, \rho_2 \, \rho_3 \, \rho_4 \, \rho_5$ and $\rho' = \rho'_5 \, \rho'_4 \, \rho'_3 \, \rho'_2 \, \rho'_1$,
as well as cycles $\pi_1, \ldots, \pi_4$ and $\pi'_1, \ldots, \pi'_4$ given by the definition of thick run.
Let $v_1, \ldots, v_4$ and $v'_1, \ldots, v'_4$
be the respective effects of $\pi_1, \ldots, \pi_4$ and $\pi'_1, \ldots, \pi'_4$.
For $j = 1, \ldots, 4$ let $c_j = \trg(\rho_j)$ and for $j = 2, \ldots, 4$ let $e_j \in\N^2$ 
be the minimal non-negative vector such that the configuration $c_j + e_j$ enables cycle $\pi_j$.
We define the following system $\U$
of linear inequalities with 6 variables $a_1, a_2, a_3, a_4, x, y$
($\max$ is understood point-wise):
\begin{align} 
a_1 v_1 & \, \geq \, e_2  \label{eq:U1} \\
a_1 v_1 + a_2 v_2 & \, \geq \, \max(e_2,e_3) \label{eq:U2} \\ 
a_1 v_1 + a_2 v_2 + a_3 v_3 & \, \geq \, \max(e_3, e_4) \label{eq:U3} \\
a_1 v_1 + a_2 v_2 + a_3 v_3 + a_4 v_4 & \, = \, (x, y) \, \geq \, e_4 \label{eq:U4} 
\end{align}
(Observe that when $v_1[j] = 0$, i.e., in case when $v_1$ is vertical or horizontal,
$e_j = 0$ and therefore one of the two first inequalities is always satisfied, namely
$a_1 v_1[j]  \, \geq \, e_2[j]$.)
Likewise, we have a system of inequalities $\U'$
with 6 variables $a'_1, a'_2, a'_3, a'_4, x', y'$.
Observe that the sequential cone 
$\poscone(v_1, v_2, v_3, v_4)$ contains exactly
(projections on $(x, y)$ of) non-negative rational solutions of the modified system 
$\U^{\zero}$ 
obtained 
by replacing all the right-hand sides
with $\zero$. 
Likewise we define ${\U'}^{\zero}$.
%
Finally, we define the compound system $\C$ by enhancing the union 
of $\U$
and $\U'$
with two additional equalities (likewise we define the system $\C^{\zero}$)
\begin{align} \label{eq:xy}
(x, y) = (x', y').
\end{align}
\begin{claim} \label{claim:Usol}
$\C$ admits a non-negative integer solution $(a_1, a_2, a_3, a_4, x, y, a'_1, a'_2, a'_3, a'_4, x', y')$.
\end{claim}
\begin{proof} 
The system $\C^{\zero}$ admits a non-negative rational solution as  
the intersection of the cones $\poscone(v_1, v_2, v_3, v_4)$ and $\poscone(v'_1, v'_2, v'_3, v'_4)$
is non-empty by assumption.
As intersection of cones is stable under multiplications by non-negative rationals,
the solution can be scaled up arbitrarily, to yield a non-negative integer one, and even a 
non-negative integer solution of the stronger system $\C$.
\end{proof}
\begin{claim} \label{claim:run}
For every non-negative integer solution 
of $\C$, for the cycles defined as $\a_j := \pi_j^{a_j}$ and $\a'_j := {(\pi'_j)}^{a'_j}$, for $j = 1, 2, 3, 4$,
the following path is a $\zero$-run:
\[
\rho_1 \; \alpha_1 \; \rho_2 \; \alpha_2 \; \rho_3 \; \alpha_3 \; \rho_4 \; \alpha_4 \; \rho_5 \; 
\rho'_5 \; \alpha'_4 \; \rho'_4 \; \alpha'_3 \; \rho'_3 \; \alpha'_2 \; \rho'_2 \; \alpha'_1 \; \rho'_1.
\]
\end{claim}
\begin{proof}
The first two inequalities~\eqref{eq:U1} enforce that the first cycle $\pi_1$ is repeated sufficiently many $a_1$ times
so that $\pi_2$ is enabled in configuration $\trg(\rho_1 \, \alpha_1 \, \rho_2)$.
Then the next two 
inequalities~\eqref{eq:U2} enforce that $\pi_1$ and $\pi_2$ are jointly repeated sufficiently many 
$a_1$, $a_2$ times so that $\pi_2$ is still enabled after its last repetition 
(which guarantees that every of intermediate repetitions of $\pi_2$ is also enabled), 
and that $\pi_3$ is enabled in configuration $\trg(\rho_1 \, \alpha_1 \, \rho_2 \, \alpha_2 \, \rho_3)$. 
Likewise for~\eqref{eq:U3}.  
Finally, the inequalities \eqref{eq:U4} enforce that 
$\pi_1, \ldots, \pi_4$ are jointly repeated sufficiently many times so that $\pi_4$ is still enabled after its last repetition.  
Analogous argument, but in the reverse order, applies for the repetitions of $\pi'_4, \ldots, \pi'_1$.
Finally, 
equalities~\eqref{eq:xy} ensure that the total effect of $\alpha_1, \ldots, \alpha_4$ is 
precisely compensated by the total effect of $\rev(\alpha'_1), \ldots, \rev(\alpha'_4)$.
\end{proof}

\begin{proof}[Proof of \cref{thm:pumping}] 
Consider a solution of $\C$. In particular the sum
$\eff(\alpha_1) + \ldots + \eff(\alpha_j)$, as well as 
$\eff(\rev(\alpha'_1)) + \ldots + \eff(\rev(\alpha'_j))$, is necessarily non-negative for every $j = 1, \ldots, 4$. 
Therefore, as a direct corollary of \cref{claim:run}, for every $i\in\N$ the path
\[
\rho_1 \, \alpha_1^i \, \rho_2 \, \alpha_2^i \, \rho_3 \, \alpha_3^i \, \rho_4 \, \alpha_4^i \, \rho_5 \, 
\rho'_5 \, (\alpha'_4)^i \, \rho'_4 \, (\alpha'_3)^i \, \rho'_3 \, (\alpha'_2)^i \, \rho'_2 \, (\alpha'_1)^i \, \rho'_1
\]
is also a $\zero$-run.
For bounding the lengths of cycles we use \cref{claim:Usol} and
apply \cref{lem:sol} to $\C$, to deduce that
$\C$ admits a non-negative integer solution of norm polynomial in $A = p(nM)^n$.
This, together with the bounds on lengths of cycles $\pi_1, \ldots, \pi_4$ and
$\pi'_1, \ldots, \pi'_4$ in the definition of $A$-thick run, 
entails required bounds on the lengths of the pumpable cycles.
Finally, the lengths of the extremal factors $\rho_1$ and $\rho'_1$ can be also bounded:
if $\rho_1$ (resp.~$\rho'_1$) is long enough it must admit a repetition of configuration, 
we add one more cycle determined by the first (resp.~last) such repetition,
thus increasing $k$ from 8 to 10.
\end{proof}

For proving \cref{thm:short-run} we will need a slightly more elaborate pumping.
By the definition of thick run, both coordinates are bounded by $A$ along $\rho_1$ and $\rho'_1$.
W.l.o.g.~assume that no configuration repeats in each of the two runs, and hence
their lengths are bounded by $A^2$.

Let $\C_\delta$ denote the union of of 
$\U$
and $\U'$
enhanced, this time, by the two equalities
\begin{align*} \label{eq:xy2}
(x, y) + (\delta_x, \delta_y) = (x', y').
\end{align*}
The two additional variables $\delta_x, \delta_y$ describe, intuitively, possible differences
between the total effect of $\pi_1^{a_1}, \ldots, \pi_4^{a_4}$ and the total effect of 
$\rev(\pi'_1)^{a'_1}, \ldots, \rev(\pi'_4)^{a'_4}$.
The projection of any solution of $\C_\delta$ on variables $(\delta_x, \delta_y)$ we call below a \emph{shift}.
%
\begin{claim} \label{claim:shift}
For some non-negative integer $m$ bounded polynomially with respect to $A$, all the four vectors
$(0, m)$, $(m, 0)$, $(0, -m)$ and $(-m, 0)$ are shifts.
\end{claim}
\begin{proof}
We reason analogously as in the proof of \cref{claim:Usol}, but this time we rely on the assumption that
intersection of the cones $\poscone(v_1, v_2, v_3, v_4)$ and $\poscone(v'_1, v'_2, v'_3, v'_4)$ is 
non-trivial, and hence contains, for some $v\in\Rpos^2$ and $a\in\Rpos$, the points $v$ and
$v + (0, a)$.
By scaling we obtain an integer point $v'\in\N^2$ and a non-negative integer $m_1\in\N$
so that $v'$ and $v' + (0, m_1)$ both belong to the intersection of cones.
Therefore the vector $(0, m_1)$ is a shift.
Likewise we obtain three other non-negative integers $m_2, m_3, m_4\in\N$
such that $(m_2, 0)$, $(0, -m_3)$ and $(-m_4, 0)$ are all shifts. 
Each of the integers $m_1, \ldots, m_4$ can be bounded polynomially in $A$ using \cref{lem:sol}.
As shifts are stable under multiplication by non-negative integers, 
it is enough to take as $m$ the least common multiple of the four integers.
\end{proof}
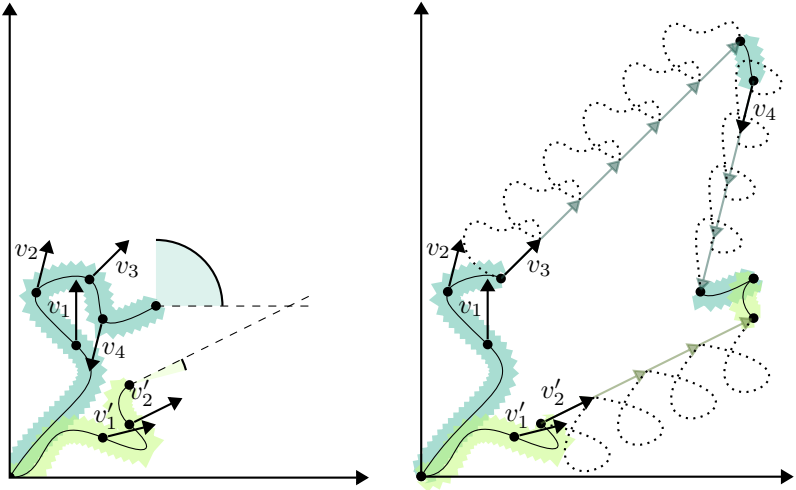
\begin{figure}[h]
    \vspace{-0.5\baselineskip}
    \centering{\input{images/08-pumping-two-steps.tikz}}
    \caption{Contracted paths $\widetilde \rho, \widetilde \rho'$ (left) 
    and reconstructed $\zero$-run $\bar \tau = \bar \rho \, \bar \rho'$ (right).}
    \label{fig:shortrun}
\end{figure}
\begin{proof}[Proof of \cref{thm:short-run}]
We use $m$ from the last claim to modify all factors of $\tau$ except for 
$\rho_1$ and 
$\rho'_1$, in order to reduce their lengths to at most $n \cdot m^2$.
W.l.o.g.~assume $m$ to be larger than $A$ (take a sufficient multiplicity of $m$ otherwise); this assumption
allows us to proceed uniformly, irrespectively whether $v_1$ is positive or not.
Observe that any path longer than $n \cdot m^2$ 
must contain two configurations 
with the same control state whose vectors are coordinate-wise congruent modulo $m$. 
As long as this happens, we remove the infix; note that this operation changes the effect of 
the whole path by a multiplicity of $m$ on every coordinate. 
If this operation is performed on factors
$
\rho_2, \,  \rho_3,  \, \rho_4, \,  \rho_5, \, \rho'_5, \, \rho'_4, \,  \rho'_3, \, \rho'_2,
$
the paths $\rho, \rho'$ are transformed into contracted paths  
(see the left picture in \cref{fig:shortrun}) of the form:
\[
\widetilde \rho  \ = \  
\rho_1 \; \widetilde \rho_2 \; \widetilde \rho_3 \; \widetilde \rho_4 \; 
\widetilde \rho_5, \qquad 
\widetilde \rho' \ = \ 
\widetilde \rho'_5 \;  \widetilde \rho'_4 \; \widetilde \rho'_3 \; \widetilde \rho'_2 \; \sigma_1,
\]
each of total length at most $5\, n \cdot m^2$.
Importantly, their effects $\eff(\widetilde \rho)$ and $\eff(\widetilde \rho')$ are bounded polynomially in $A$,
and their difference is (coordinate-wise) divisible by $m$:
\[
\eff(\widetilde \rho) - \eff(\rev(\widetilde \rho')) \ = \ (a m, b m)
\qquad \text{for some integers $a, b \in\Z$ polynomial in $A$.}
\]
Our aim is now to pump up the cycles $\pi_1, \ldots, \pi_4$ and $\rev(\pi'_1), \ldots, \rev(\pi'_4)$ 
(see the right picture in \cref{fig:shortrun}), 
to finally end up with the paths of the form
\begin{align} \label{eq:runagain}
\bar \rho \ = \ \rho_1 \, 
\pi_1^{a_1} \, \widetilde \rho_2 \, \pi_2^{a_2} \, \widetilde \rho_3 \, \pi_3^{a_3} \, \widetilde \rho_4 \, \pi_4^{a_4}
\, \widetilde \rho_5, \quad
\bar \rho' \ = \ \widetilde \rho'_5 \, 
(\pi'_4)^{a'_4} \, \widetilde \rho'_4 \, (\pi'_3)^{a'_3} \, \widetilde \rho'_3 \, (\pi'_2)^{a'_2} \, \widetilde \rho'_2 \, (\pi'_1)^{a'_1}
 \, \rho'_1,
\end{align}
such that $\bar \tau = \bar \rho \, \bar \rho'$ is a $\zero$-run.
In other words, we aim at $\eff(\bar \rho) = \eff(\rev(\bar \rho'))$.
We are going to use Lemma~\ref{lem:sol} twice.
For $j = 2, \ldots, 5$ let $c_j := \eff(\rho_1 \widetilde \rho_2 \ldots \widetilde \rho_j) \in\Z^2$, 
and let $f_j$ be the minimal non-negative vector such that the configuration 
$c_{j-1} + f_j$ enables $\widetilde \rho_{j}$.
For $j = 2, \ldots, 4$ let $e_j\in\N^2$ be the minimal non-negative vector such that the configuration 
$c_j + e_j$ enables $\pi_j$. Finally, let $e_5$ be the minimal non-negative vector such that $c_5 + e_5 \geq \zero$.
Analogously to the system $\U$~\eqref{eq:U1}--\eqref{eq:U4},
we define the system $\widetilde \U$ of linear inequalities:
%
\begin{align*} 
\begin{aligned}
a_1 m v_1  \ & \geq \ \max(e_2, f_2) \\
a_1 m v_1 + a_2 m v_2  \ & \geq \ \max(e_2, e_3, f_3) \\
a_1 m v_1 + a_2 m v_2 + a_3 m v_3  \ & \geq \ \max(e_3, e_4, f_4) \\
a_1 m v_1 + a_2 m v_2 + a_3 m v_3 + a_4 m v_4  \ & \geq \ \max(e_4, e_5, f_5) 
\end{aligned}
\end{align*}
%
%
In words, $\widetilde \U$ requires that every prefix of $\bar \rho$ is enabled in the source $\zero$-configuration,
and that the number of repetitions of every cycle $\pi_i$ is divisible by $m$.
Clearly $\widetilde \U$ has a non-negative integer solution, as
$v_1$ is either positive, or vertical or horizontal in which case $v_2$ is positive on the relevant coordinate.
Likewise we define a system of inequalities $\widetilde \U'$ that 
requires that every prefix of $\rev(\bar \rho')$ is enabled in the target $\zero$-configuration.
%
%
Consider some fixed solutions of $\widetilde \U$ and $\widetilde \U'$ bounded,
by the virtue of \cref{lem:sol}, polynomially in $A$.
We have thus two fixed runs $\bar \rho$ and $\rev(\bar \rho')$
of the form~\eqref{eq:runagain}, with source vector $\zero$;
the number of repetitions of each cycles is divisible by $m$,
and the difference of their effects is (coordinate-wise) divisible by $m$:
\[
\eff (\bar \rho) - \eff(\rev(\bar \rho')) \ = \ (am, bm)
\qquad \text{for some integers $a, b \in\Z$ polynomial in $A$.}
\]
As shifts are closed under addition, by \cref{claim:shift} we know that $(am, bm)$ is a shift.
Substituting $(am, bm)$ for $(\delta_x, \delta_y)$ in the system $\C_\delta$ yields a system which admits,
again by \cref{lem:sol}, a solution bounded polynomially in $A$.
We use such a solution to increase the numbers of repetitions of respective cycles 
$a_1, \ldots, a_4$ and $a'_4, \ldots, a'_1$ in $\bar \rho$ and $\bar \rho'$, respectively.
This turns the path $\bar \tau = \bar\rho \, \bar\rho'$ into a $\zero$-run of length bounded polynomially in $A$.
\end{proof}

%% file: images/08-pumping-two-steps.tikz
\begin{tikzpicture}
    \node[anchor=east] (S) at (0,0) {\input{images/08-B-pumping.tikz}};
    \node[anchor=west] (T) at (0.2,0)  {\input{images/08-C-pumping.tikz}};
\end{tikzpicture}

%% file: images/08-B-pumping.tikz
\begin{tikzpicture}[anchor=center,scale=0.35]
    \drawGridB{13.5}{18}

    \clip (0,0) rectangle (11.5,18);
    \begin{scope}[yshift=16cm]
    \begin{scope}[xscale=0.05,yscale=-0.05]

        \begin{scope}[even odd rule,xshift=-120cm,yshift=50cm]
            \clip (320,140) -- (320,0) -- (230,-10) -- (230,140) -- cycle;
            \draw[sequentialConeAngleA] (230,140) circle (50);
        \end{scope}
        \begin{scope}[even odd rule,xshift=-100cm,yshift=-50cm]
            \clip (380,250) -- (370,210) -- (190,300) -- cycle;
            \draw[sequentialConeAngleB] (190,300) circle (44);
        \end{scope}  

        \draw[runOutlineA,line width=2.5mm,decoration = {zigzag,segment length=1mm, amplitude=0.25mm},decorate]    (110,190)
            .. controls (100.5,196) and (74.5,209) .. (70,200)
            .. controls (65.5,191) and (69.5,177) .. (60,170)
            .. controls (50.5,163) and (21.5,172) .. (20,180)
            .. controls (18.5,188) and (16.5,188) .. (50,220)
            .. controls (83.5,252) and (34.5,264) .. (0,320) ;
        \draw[runOutlineB,line width=2.5mm,decoration = {zigzag,segment length=1mm, amplitude=0.25mm},decorate]    (0,320)
            .. controls (34.5,321) and (16.5,267) .. (70,290)
            .. controls (123.5,313) and (97.5,286) .. (90,280)
            .. controls (82.5,274) and (76.5,259) .. (90,250) ;

        \draw[runLineBlack]    (110,190) node[point] {}
            .. controls (100.5,196) and (74.5,209) .. (70,200) node[point] (p4) {}
            .. controls (65.5,191) and (69.5,177) .. (60,170) node[point] (p3) {}
            .. controls (50.5,163) and (21.5,172) .. (20,180) node[point] (p2) {}
            .. controls (18.5,188) and (16.5,188) .. (50,220) node[point] (p1) {}
            .. controls (83.5,252) and (34.5,264) .. (0,320) node[point] {} ;
        \draw[runLineBlack]    (0,320) node[point] {}
            .. controls (34.5,321) and (16.5,267) .. (70,290) node[point] (q1) {}
            .. controls (123.5,313) and (97.5,286) .. (90,280) node[point] (q2) {}
            .. controls (82.5,274) and (76.5,259) .. (90,250) node[point] {} ;

        \draw[vector]    (q2) -- node[vectorLabel,pos=0.5] {$v_2'$} +(230-190,280-300) ;
        \draw[vector,yshift=42.5cm]    
                            (p1)  -- node[vectorLabel,pos=0.5] {$v_1$}  +(50-50,170-220) ;
        \draw[vector]    (p2)  -- node[vectorLabel,pos=0.5] {$v_2$}  +(70-60,70-110) ;
        \draw[vector]    (p4) -- node[vectorLabel,pos=0.3] {$v_4$}  +(-10,180-140) ;
        \draw[vector]    (p3)  -- node[vectorLabel,pos=0.5,swap] {$v_3$}  +(200-170,40-70) ;
        \draw[vector]    (q1)  -- node[vectorLabel,pos=0.3] {$v_1'$} +(110-70,280-290) ;
    \end{scope}
    \end{scope}
    \draw[dashed,thin] (5.5,6.5) -- (11.5,6.5);
    \draw[dashed,thin] (4.5,3.5) -- ++(7,3.5);

\end{tikzpicture}

%% file: images/08-C-pumping.tikz
\begin{tikzpicture}[anchor=center,scale=0.35]
    \drawGridB{14}{18}

    \begin{scope}[yshift=16cm]
    \begin{scope}[xscale=0.05,yscale=-0.05]

        \draw[runOutlineA,line width=2.5mm,decoration = {zigzag,segment length=1mm, amplitude=0.25mm},decorate,xshift=140cm,yshift=-20cm]    (110,190) 
            .. controls (100.5,196) and (74.5,209) .. (70,200) ;

        \draw[runOutlineA,line width=2.5mm,decoration = {zigzag,segment length=1mm, amplitude=0.25mm},decorate,xshift=180cm,yshift=-180cm]  (70,200)  
            .. controls (65.5,191) and (69.5,177) .. (60,170) ;
        
        \draw[runOutlineA,line width=2.5mm,decoration = {zigzag,segment length=1mm, amplitude=0.25mm},decorate]  (60,170) 
            .. controls (50.5,163) and (21.5,172) .. (20,180) 
            .. controls (18.5,188) and (16.5,188) .. (50,220) 
            .. controls (83.5,252) and (34.5,264) .. (0,320) ;
        \draw[runOutlineB,line width=2.5mm,decoration = {zigzag,segment length=1mm, amplitude=0.25mm},decorate]    (0,320) 
            .. controls (34.5,321) and (16.5,267) .. (70,290) 
            .. controls (123.5,313) and (97.5,286) .. (90,280) ;
        \draw[runOutlineB,line width=2.5mm,decoration = {zigzag,segment length=1mm, amplitude=0.25mm},decorate,xshift=160cm,yshift=-80cm]     (90,280) 
            .. controls (82.5,274) and (76.5,259) .. (90,250) ;

        \draw[runLineBlack,xshift=140cm,yshift=-20cm]    (110,190) node[point] {}
            .. controls (100.5,196) and (74.5,209) .. (70,200) node[point] (p4up) {} ;

        \draw[runLineBlack,xshift=180cm,yshift=-180cm]  (70,200)  node[point] (p4down) {}
            .. controls (65.5,191) and (69.5,177) .. (60,170) node[point] (p3up) {} ;
        
        \draw[runLineBlack]  (60,170) node[point] (p3down) {}
            .. controls (50.5,163) and (21.5,172) .. (20,180) node[point] (p2) {}
            .. controls (18.5,188) and (16.5,188) .. (50,220) node[point] (p1) {}
            .. controls (83.5,252) and (34.5,264) .. (0,320) node[point] {} ;
        \draw[runLineBlack]    (0,320) node[point] {}
            .. controls (34.5,321) and (16.5,267) .. (70,290) node[point] (q1) {}
            .. controls (123.5,313) and (97.5,286) .. (90,280) node[point] (q2up) {} ;
        \draw[runLineBlack,xshift=160cm,yshift=-80cm]     (90,280) node[point] (q2down) {}
            .. controls (82.5,274) and (76.5,259) .. (90,250) node[point] {} ;

        \foreach \i in {0,...,5} {
            \draw[basicArrow,draw=cA!50!black,opacity=0.5] (p3down) ++(30*\i,-30*\i) -- ++(30,-30);
            \draw [dotted,thick,color=black,draw opacity=1 ]  (p3down) ++(30*\i,-30*\i) .. 
                controls +(128.5-170,58-70) and +(140.5-160,15-30) .. ++(160-170,30-70) .. 
                controls +(179.5-160,45-30) and +(196.5-200,6-40) .. ++(200-160,40-30) ;

        }

        \foreach \i in {0,...,3} {
            \draw[basicArrow,draw=cA!50!black,opacity=0.5] (p4down) ++(-10*\i,40*\i) -- ++(-10,40);
            \draw [dotted,thick,color=black,draw opacity=1] (p4down) ++(-10*\i-10,40*\i+40) 
             .. controls +(219.5-220,154-180) and +(205.5-240,109-130) .. ++(240-220,130-180)
             .. controls +(274.5-240,151-130) and +(216.5-230,158-140) .. ++(230-240,140-130) ;
            \draw[basicArrow,draw=cB!50!black,opacity=0.5] (q2up) ++(40*\i,-20*\i) -- ++(40,-20);
            \draw [dotted,thick,color=black,draw opacity=1 ] (q2up) ++(40*\i+40,-20*\i-20)
                .. controls +(191.5-230,324-280) and +(211.5-230,346-340) .. ++(230-230,340-280-5)
                .. controls +(248.5-230,334-340) and +(251.5-190,311-300) .. ++(190-230,300-340+5) ;

        }

        \draw[vector]    (q2up) -- node[vectorLabel,pos=0.5] {$v_2'$} +(230-190,280-300) ;
        \draw[vector,yshift=42.5cm]    
                            (p1)  -- node[vectorLabel,pos=0.5] {$v_1$}  +(50-50,170-220) ;
        \draw[vector]    (p2)  -- node[vectorLabel,pos=0.5] {$v_2$}  +(10,-40) ;
        \draw[vector]    (p4down) -- node[vectorLabel,pos=0.3] {$v_4$}  +(-10,180-140) ;
        \draw[vector]    (p3down)  -- node[vectorLabel,pos=0.5,swap] {$v_3$}  +(30,-30) ;
        \draw[vector]    (q1)  -- node[vectorLabel,pos=0.3] {$v_1'$} +(110-70,280-290) ;
    \end{scope}
    \end{scope}

\end{tikzpicture}

%% file: first-cycle.tex

\section{Proof of Non-negative Cycle Lemma}  

In this section we prove \cref{lem:first-cycle}.
Fix a \twovass $V$ with $n$ states, and let $M = \norm{V}$.
We proceed by a sequence of auxiliary lemmas. 

\begin{wrapfigure}[17]{r}[0cm]{4cm}
  \centering
  \vspace{-\baselineskip}
  \input{images/06-repeated-state.tikz}
  \label{fig:increasing-angle}
\end{wrapfigure}

\begin{lemma}\label{lem:one-coordinate-bounded}
Let $\rho$ be a run such that one of coordinates is smaller than $K$ in all configurations in $\rho$,
and such that $\norm{\trg(\rho)} > \norm{\src(\rho)} + K n M$. 
Then 
\begin{enumerate}
\item [{\sc (i)}] $\rho$ contains, as an infix, a cycle with vertical or horizontal effect,
\item [{\sc (ii)}] $\rho$ contains a configuration enabling such a cycle of length polynomial in $K n M$.    
\end{enumerate}
\end{lemma}

\subparagraph{Proof.}
W.l.o.g.~assume that the first (horizontal) coordinate is bounded by $K$ in all configurations in $\rho$.
Let $s = \src(\rho)$ and $t = \trg(\rho)$.

We first prove that $\rho$ contains a cycle with vertical effect.  
Define a sequence of configurations $c_0, c_1, \ldots, c_m$ as follows.
Let $c_0$ be the first configuration which minimizes the value of the second (vertical) coordinate;
clearly $c_0[2] \leq s[2]$.
Further, let $c_{i+1}$ be the first configuration in $\rho$ such that $c_{i+1}[2] > c_i[2]$.
Thus $c_{i+1}[2] \leq c_i[2] + M$, and in consequence 
\[
t[2] \leq c_m[2] \leq c_0[2] + mM \leq s[2] + mM.
\]
According to the assumption we have $\norm{t} > \norm{s} + K n M$ hence, as the first coordinate is bounded by $K$,
we deduce the inequality 
\[
t[2] > s[2] + K n M.
\]
The two above inequalities relating $t[2]$ and $s[2]$ imply $m > K n$.
Therefore there must be two configurations $c_i$ and $c_j$, for $0 \leq i < j \leq K n$,
with the same control state $q$ and the same first coordinate $c_i[1] = c_j[1]$, 
and thus the infix $\rho_{ij}$ of $\rho$ from $c_i$ to $c_j$
is a cycle with effect $(0, y)$, where $0 < y \leq (j-i) M \leq K n M$.

Now we bound the length of the cycle.
For all configurations in $\rho_{ij}$, we observe that
the first coordinate stays between $0$ and $K-1$, and the second coordinate stays between
$c_0[2]$ and $c_j[2]$. 
We know that $j \leq K n$, hence $c_j[2] \leq c_0[2] + K n M$.
In consequence, the counter values
in all configurations in the cycle $\rho_{ij}$ are restricted to at most $K (K n M + 1)$ different vectors,
and therefore there are at most $L = K n (K n M + 1)$ different configurations in $\rho_{ij}$.
By removing repetitions of configurations, i.e., by removing cycles of effect $\zero$, we reduce the length
of the cycle to at most $L$, which is bounded polynomially in $K n M$.
\qed
\vspace{\baselineskip}

Every \twovass $V$ induces a directed graph whose vertices are control states of $V$, 
with an edge from $p$ to $q$ if and only 
if $V$ has a transition of the form $(p, v, q)$. This graph allows us to split control states of $V$ into
strongly connected components, which we call briefly {\SCC}s.
The following lemma distinguishes two kinds of {\SCC}s:
\begin{lemma}\label{lem:scc-dichotomy}
Every \SCC $S$ satisfies one of the following conditions:
\begin{enumerate}
  \item[(a)] every control state in $S$ belongs to some positive cycle of length polynomial in $nM$; 
  \item[(b)] the effects of all cycles in $S$ belong to some half-plane   
  containing no positive vector.
\end{enumerate}
\end{lemma}

\begin{proof}
%
Let $U$ be the set of effects of simple cycles included in $S$. We consider two cases:

\subparagraph{Case 1: $\cone(U)$ contains a positive vector.}
Fix an arbitrary positive vector $v \in \cone(U)$. By Caratheodory's Theorem, 
$v = a_1 u_1 + a_2 u_2 \in \cone(u_1, u_2)$ for some two vectors $u_1, u_2 \in U$ and $a_1, a_2 \in \N$. 
By \cref{lem:sol} we know that $a_1 u_1 + a_2 u_2$ is positive for some non-negative integers 
$\a_1, \a_2 \leq (2 M)^2$. 
We also know that $u_1$ is the effect of a simple cycle $\pi_1$ from, say,  state $q_1$ to $q_1$;
and $u_2$ is the effect of a simple cycle $\pi_2$ from state $q_2$ to $q_2$.

Fix a state $q \in S$. As $S$ is strongly connected it contains a cycle $\pi$ of length at most $3n$
which contains all $q$, $q_1$ and $q_2$.
Thus absolute values of $\eff(\pi)$ on both coordinates are at most $3nM$, hence are larger or equal than $-3nM$. 
Therefore $\pi$, together with cycle $\pi_1$ repeated $a_1 \cdot (3nM + 1)$ times, and with
cycle $\pi_2$ repeated $a_2 \cdot (3nM + 1)$ times, form a cycle with positive effect.
The length of this cycle is at most $3n + 2n(2M)^2 (3nM + 1)$, hence bounded polynomially in $n M$.
Condition (a) holds.

\subparagraph{Case 2: $\cone(U)$ contains no positive vector.}
By \cref{lem:cone-dichotomy} we deduce that $\cone(U)$ is included in some
half-plane $\Pi$. If $\Pi$ intersects the positive orthant $\Rpos^2$, rotate the half-plane so that it is disjoint
from $\Rpos^2$. The so obtained half-plane $\Pi'$ contains no positive vector and still includes $\cone(U)$,
hence condition (b) holds.
\end{proof}

\begin{lemma}\label{lem:first-cycle-scc}
There is a polynomial $Q$ such that every run $\rho$ within one \SCC with
$\norm{\trg(\rho)} > Q(n M) \cdot (\norm{\src(\rho)} + 1)$
contains a configuration enabling a semi-positive cycle of length 
at most $Q(n M)$.
\end{lemma}
\begin{proof}
Let $Q_1$ and $Q_2$ be the polynomials from \cref{lem:scc-dichotomy}{\sc (ii)}
and \cref{lem:one-coordinate-bounded}(a),
respectively.
Let $s = \src(\rho)$ and $t = \trg(\rho)$, and let $S$ be the \SCC containg $\rho$. 
We split the proof according to the two cases (a) and (b) of \cref{lem:scc-dichotomy}.
The proof goes through for every polynomial $Q$ satisfying the following two inequalities:
\begin{align*}
Q(x) \ & \geq \ Q_1(Q_2(x) \cdot x^2)   & \text{(Case 1)}\\
Q(x) \ & \geq \ x^2  & \text{(Case 2)}
\end{align*}
%
%
%

\subparagraph{Case 1: $S$ satisfies (a).}

If $\rho$ visits some configuration with both coordinates at least 
$Q_2(n M) \cdot M = K$
then this configuration necessarily enables a positive cycle of length bounded by $Q_2(n M) \leq Q(n M)$.
Otherwise, we know that in every configuration in $\rho$ one of coordinates is smaller than $K$.
W.l.o.g.~assume $t[1] < K$.
Let $\rho'$ be the longest suffix of $\rho$ such that the first coordinate is bounded by $K-1$ along $\rho'$, and let 
$s' = \src(\rho')$.
We claim that $\norm{s'} \leq \norm{s} + K -1 + M$; indeed, if $s' \neq s$, 
the first coordinate of the configuration $u$ preceding $s'$ in $\rho$ is at least $K$,
and therefore $u[2] \leq K-1$, which implies that $s'[2] \leq K-1 + M$.

By assumption we know that $\norm{t} > Q(n M) \cdot (\norm{s} + 1)$, and hence necessarily
$\norm{t} > \norm{s'} + K n M$.
We can thus apply \cref{lem:one-coordinate-bounded}{\sc (ii)} to $\rho'$, to learn that some configuration in $\rho'$
enables a vertical cycle of length at most $Q_1(K n M) \leq Q(n M)$.

\subparagraph{Case 2: $S$ satisfies (b).}

Denoting by $U$ the set of all simple cycles in $S$, due to condition (b) we know that 
$\cone(U)$ is included in some half-plane
$\Pi = \ccangle{-w}{w}$, where $-w \in \N\times (-\N)$ and $w\in (-\N)\times\N$. 
We aim at showing the following claim:

\begin{claim} \label{claim:Uvert}
$U$ contains a vertical or horizontal cycle.
\end{claim}
Towards contradiction suppose $U$ contains no vertical nor horizontal cycle.
Whenever a vector $p = (-x, y) \in (-\N) \times \N$, for $y > x > 0$,  
is the effect of a simple cycle, its \emph{ratio} $y/x$ is necessarily bounded by $n M$.
Therefore, the vector $w$ determining $\Pi$ can be assumed to have ratio bounded by $n M$ as well.
Note that all cycles contained as an infix in $\rho$, necessarily belong to $\Pi$. 
We are going to show bounds on $t[1]$ and $t[2]$ which contradict the assumption on $\norm{t}$.

Factor the run $\rho$ into a at most $n$ (not necessarily simple) cycles, interleaved with at most 
$n-1$ remaining transitions.
Thus we have $t = s + r + p$, where $p\in\Pi$ is the total effect of the cycles and $r$ is the total effect of at most $n-1$ transitions.
Let $p'$ denote the total effect of those among the cycles whose vertical effect is non-negative
(and hence horizontal effect is forcedly negative). 
Thus
\[
t[2] \ \leq \ (s + r + p')[2].
\]
As the half-plane $\Pi = \ccangle{-w}{w}$ contains all these cycles, and the ratio of $w$ is bounded by $nM$
as discussed above, we know that ratio of $p'$ is also bounded by $nM$.
In consequence
$p'[2] \ \leq \ -p'[1] \ \leq \ (s + r)[1] \cdot n M$, and hence
$$t[2] \ \leq \ (\norm{s} + \norm{r}) \cdot (1 + n M) \ \leq \ (\norm{s} + (n-1) M) \cdot (1 + n M).$$
As the same bound is obtained symmetrically for $t[1]$, we have arrived at a contradiction with the assumption
$\norm{t} > Q(n M) \cdot (\norm{s} + 1)$.
\Cref{claim:Uvert} is thus proved.

\begin{claim} \label{claim:Pi}
The run $\rho$ contains, as an infix, a vertical or horizontal cycle $\pi$.
\end{claim}
W.l.o.g.~supose $U$ contains a vertical cycle. In consequence, no cycle in $S$
has positive first (horizontal) coordinate. Therefore the horizontal coordinate 
is smaller than $K = s[1] + (n-1)M + 1$ in all configurations in $\rho$.
By \cref{lem:one-coordinate-bounded}{\sc (i)} $\rho$ contains, as an infix, a vertical cycle.

Relying on the \cref{claim:Pi}, w.l.o.g.~assume $\rho$ contains a vertical cycle $\pi$ as infix.
For completing the proof of \cref{lem:first-cycle-scc} we need to bound the length of $\pi$.
As $S$ satisfies condition (b), it contains no cycle with positive horizontal effect; in consequence,
decomposition of $\pi$ into simple cycles uses only cycles with effect $(0, a)$, where $a\in\Z$.
Split these simple cycles into \emph{increasing} ($a > 0$) and \emph{non-increasing} ($a \leq 0$).
Suppose the length of $\pi = \pi_0$ 
is greater than $n$ and consider the first simple cycle $\sigma_1$ contained as its infix.
If $\sigma_1$ is non-increasing remove $\sigma_1$ from $\pi$, thus obtaining the path $\pi_1$, 
and consider the first simple cycle $\sigma_2$ contained in $\pi_1$ as an infix. 
Again, remove $\sigma_2$ if it is non-increasing. And so on, continue this process until finally
certain cycle $\sigma_i$ in $\pi_{i-1}$ is increasing.
As all the removed simple cycles $\sigma_1, \ldots, \sigma_{i-1}$ were non-increasing, 
inserting back to $\pi_{i-1}$ those of them which preceed $\sigma_i$ necessarily increases the configuration 
$\src(\sigma_i)$ in $\pi_{i-1}$ so that it enables $\sigma$. 
The proof is thus completed.
\end{proof}

\begin{proof}[Proof of~\cref{lem:first-cycle}]
Let $Q$ be the polynomial from \cref{lem:first-cycle-scc}.
We define a polynomial $P(x) = Q(x) \cdot (x+1)$.
Consider a run $\rho$ from a $\zero$-configuration to some target configuration $t$. 
Let $k\leq n$ be the number of {\SCC}s traversed by the run $\rho$ and, for $i = 1, \ldots, k$, let
$s_i$ and $t_i$ be the first and the last configuration in the $i$-th \SCC, respectively.
Then $s_1 = \zero$ and $t_k = t$.
Suppose, towards contradiction, that 
$\rho$ contains no configuration enabling a semi-positive cycle of length at most $P(n M)$.
As $Q(n M) \leq P(n M)$, by \cref{lem:first-cycle-scc} we obtain
\begin{equation}\label{eq:single-scc}
\norm{t_i} \leq Q(n, M) \cdot (\norm{s_i} + 1)
\end{equation}
for $i = 1, \ldots, k$. We show by induction on $i$ that $\norm{t_i} \leq P(n M)^i$.
For $i = 1$ we use~\eqref{eq:single-scc} and the equality $\norm {s_1} = 0$, to obtain 
$\norm{t_1} \leq Q(n M) \leq P(n M)$.
For the induction step we use~\eqref{eq:single-scc} and the inequality $\norm{s_{i+1}} \leq \norm{t_i} + M$,
to obtain:
\begin{align*}
\norm{t_{i+1}}  \leq \, & Q(n M) \cdot (\norm{s_{i+1}} + 1) \leq \\
& Q(n M) \cdot (\norm{t_i} + M + 1) \leq \\
&  Q(n M) \cdot (P(n M)^i + M + 1) 
\leq P(n M)^{i+1},
\end{align*}
as required. 
Thus $\norm{t} \leq P(n M)^n$ which contradicts the assumption on $\norm{t}$ 
and therefore completes the proof. 
\end{proof}

%% file: images/06-repeated-state.tikz
\begin{tikzpicture}[scale=0.5]
    \drawGridB{5}{12.5}
    \node at (4,-0.6) {$K$};
    \fill[cA,fill opacity=0.3] (0,0) rectangle (4,12.5);
    \foreach \i in {1,2,...,7} {
        \draw[white,semithick,draw opacity=0.7] (\i/2,0) -- (\i/2,12.5);
    }
    \draw[black,thick,dashed] (4,-0.2) -- (4,12.5);

    \begin{scope}[yshift=16cm]
    \begin{scope}[xscale=0.05,yscale=-0.05]
        \draw[runLineBlack] (40,90) node[point,label=left:{$q$}] (b) {}
            .. controls (55.5,91)  and (71.5,109) .. (60,130) node[point,label=right:{$r$}] {}
            .. controls (48.5,151) and (22.5,135) .. (20,170) node[point,label=left:{$s$}] {}
            .. controls (17.5,205) and (80.5,233) .. (70,200) node[point,label=left:{$r$}] {}
            .. controls (59.5,167) and (19.5,219) .. (40,250) node[point,label=left:{$q$}] (a) {}
            .. controls (60.5,281) and (53.5,314) .. (0,320)  node[point,label=left:{$p$}] {};
    \end{scope}
    \end{scope}

    \draw[vector] (a) -- (b);
\end{tikzpicture}

%% file: proof-reach-thin.tex
\newpage
\section{Missing proof from \cref{sec:app} -- the case of thin run}

As usual we use $n = |Q|$ for the number of control states, and $M= \norm{V}$ for the norm of $V$.
Assume a $\zero$-run $\tau$ to be $A$-thin:
every configuration in $\tau$ lies in some $A$-belt $\mathcal{B}_{v,W}$. 
%
Fix $W =  A + \sqrt{2}M$ and $B = 6 W A^2 + 3W$, and let $S = [0,B]^2$.
Let $\norm{v}_2$ denote the Euclidean norm of $v$. Note that $\norm{v}_2 \leq \sqrt{2}\norm{v}$. 

\begin{claim}\label{claim:independent-belts}
	The run $\tau$ does not change belts outside of $S$, i.e., any two consecutive configurations 
	$(q, w)$, $(q', w')$ in $\tau$ 
	satisfying $w, w' \notin S$ 
	share a common belt.
\end{claim}

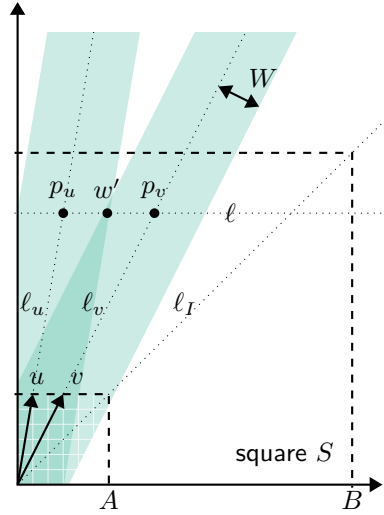
\begin{wrapfigure}[20]{r}[0cm]{5.2cm}  
	\centering
	\vspace{-\baselineskip}
	\input{images/07-disjoint-belts.tikz}
	\caption{$A$-belts intersect only within square $S$.}\label{fig:disjoint-belts}
\end{wrapfigure}

\subparagraph*{Proof.}
	Assume that $w \in \belt{u}{A}$ for some $u$ ($\norm{u} \leq A$). 
	We will show $w' \in \belt{u}{A}$. 
	Notice that $w' \in \belt{u}{W}$ 
	(since $\norm{w' - w}_2 \leq \sqrt{2}\norm{w' - w} \leq \sqrt{2}M$). 
	Towards contradiction assume that $v'$ also belongs to some $A$-belt $\belt{v}{A} \neq \belt{u}{A}$ 
	(i.e. $v$ and $u$ non-colinear). Then of course $w' \in \belt{v}{W}$ too. 
	We will show that this implies $w' \in S$.

W.l.o.g.~assume that $u \clockwise v$.
Let $I = (1,1)$. Notice that when $u \clockwise I \clockwise v$ then $w'$ also belongs to $\belt{I}{W}$. 
Thus we can assume that $u \clockwise v \clockwise I$ or $I \clockwise u \clockwise v$. 
W.l.o.g.~let us choose the first option.
Note that this implies that $u[2],v[2] > 0$.

Let $p_u$ and $p_v$ be the intersection points of $\ell_{u}$ and $\ell_{v}$ with the horizontal line $\ell : y = w'[2]$. 
Their horizontal coordinates are $\frac{u[1]}{u[2]} \cdot w'[2]$ and $\frac{v[1]}{v[2]} \cdot w'[2]$, respectively, so 
\[\norm{p_u - p_v} = w'[2]\frac{|u[1]v[2] - v[1]u[2]|}{u[2]v[2]}.\]
 Because the belts intersect with $\ell$ at an angle between $45^\circ$ and $90^\circ$, 
the line segments $\belt{u}{W} \cap \ell$ and $\belt{v}{W} \cap \ell$ are of length $\leq 2\sqrt{2} W < 3W$. 
Thus $\norm{p_u - p_v} \leq \norm{p_u - w'} + \norm{p_v - w'} < 6W$. Consequently:
\begin{align*}
	&w'[2] \frac{|u[1]v[2] - v[1]u[2]|}{u[2]v[2]} \ < \ 6W \\
	&w'[2] < 6W \frac{u[2]v[2]}{|u[1]v[2] - v[1]u[2]|} \ < \ 6W A^2  \ \leq \ B - 3W \ < \ B.
\end{align*}
Furthermore $p_v[1] < p_v[2] = w'[2]$ and $\norm{p_v - w'} < 3W$ so $w'[1] < B$ too, contradiction.
\qed

\newcommand{\limitNumberOfBlocks}{(A^2\cdot n)^2}
\begin{claim} \label{claim:2loops}
	Let $C = \limitNumberOfBlocks$. If $\tau$ visits a configuration of norm 
	larger than $D = B + C \cdot A$, then it 
	decomposes into 
	$
	\tau = \tau_0 \,\alpha_1\, \tau_1 \, \alpha_2\,\tau_2,
	$
	 for two cycles $\alpha_1, \alpha_2$ of opposite effects $\eff(\alpha_1) = - \eff(\alpha_2) \geq \zero$ 
	 containing jointly at most $C (C+1)$ different configurations.
\end{claim}
\begin{proof}
	For a configuration $c$ of norm larger than $D$, 
	let us decompose $\tau$ into 
	\[
	\tau = \pi\,\gamma\, \gamma' \, \pi'
	\] 
	such that $\trg(\gamma)=c = \src(\gamma')$ and 
	$\gamma\gamma'$ is a maximal infix of $\pi$ that visits only configurations of norm greater than $B$.
	By \cref{claim:independent-belts}, there exists unique belt $\mathcal{B} = \belt{u}{A}$ that contains 
	$\gamma\gamma'$.
	Assume w.l.o.g.~that $\norm{u} > M$. 
	Let us divide $\mathcal{B} \setminus S$ into segments $\mathcal{B}_i$ as follows:
	\begin{equation*}
		\begin{array}{ll}
		S_i           := \left[0,B + i u[1]\right]\times\left[0,B + i u[2]\right] 
		&\hspace{15mm}
		\mathcal{B}_i := \mathcal{B} \cap (S_{i+1} \setminus S_i).
		\end{array}
	\end{equation*}
	Observe that $\gamma\gamma'$ visits more than $C$ initial blocks $\mathcal{B}_i$ starting from $\mathcal{B}_0$ 
	up to $\mathcal{B}_{C}$, since
	a single transition cannot `jump' over a block without visiting it. 
	Let $c_i = (q_i, v_i)$ be the first configuration in $\gamma$ beloging to $\mathcal{B}_i$, and symmetrically
	let $c'_i = (q'_i, v'_i)$ be the last configuration in $\gamma'$ beloging to $\mathcal{B}_i$.
%
%
	Observe that each block $\mathcal{B}_i$ has the same shape as $\mathcal{B}_0$ 
	and differs only by translation by $i u$. 
	Furthermore, as $\norm{u} \leq A$, each $\mathcal{B}_i$ fits inside a square of size $A$   
	so it contains at most $A^2$ points.
	By the pigeonhole principle, there are at least two $i, j$ ($0 \leq i < i+d = j \leq C$) such that 
	\[
	q_i = q_j \qquad v_i + d u = v_j \qquad
	q'_i = q'_j \qquad v'_i + d u = v'_j.
	\]
	Taking as $\alpha_1$ the infix from $c_i$ to $c_j$, and as $\alpha_2$ the infix from $c'_j$ to $c'_i$, we
	obtain two required cycles.
%
\end{proof}

\begin{claim} \label{claim:2loopsagain}
	Under assumption of \cref{claim:2loops}, $\tau$ decomposes into 
	$
	\tau = \tau_0 \,\alpha_1\, \tau_1 \, \alpha_2\,\tau_2
	$
	so that $\tau_0 \, \tau_1 \, \tau_2$ is also an $A$-thin $\zero$-run.
\end{claim}
\begin{proof}
The same proof as for \cref{claim:2loops}, with one modification:
take as $c_i$ \emph{the last} configuration in $\gamma$ belonging to $\mathcal{B}_i$, and symmetrically
take as $c'_i$ \emph{the first} configuration in $\gamma'$ belonging to $\mathcal{B}_i$. 
\end{proof}

\begin{proof}[Proof of \cref{thm:pumping}]
Applying \cref{claim:2loops} simultaneously to the first belt in which the norm $D$ is exceeded, 
and to the very last such belt,
we get $\zero$-runs
\[
\tau_0 \, \alpha_1^i \, \tau_1 \, \alpha_2^i \, \tau_2
		\alpha_3^i \, \tau_3 \, \alpha_4^i \, \tau_4,
\]
for $i\in\N$, where cycles $\alpha_1, \alpha_2$ belong to the first belt and
cycles $\alpha_3, \alpha_4$ belong to the last one.
The lengths of the cycles can be reduced to at most $C(C+1)$ by removing repetitions of configurations.
Then the length of the very first factor $\tau_0$ can be bounded by $(D+1)^2 + C(C+1)$
by replacing, if needed, cycles $\alpha_1, \alpha_2$ with the first cycle of effect $\zero$ in $\tau_0$.
Likewise for the very last factor $\tau_4$.
\end{proof}

\begin{proof}[Proof of \cref{thm:short-run}]
Immeediate using \cref{claim:2loopsagain}, according to which every $A$-thin $\zero$-run exceeding 
norm $D$ can be shortened.
Once all configurations along a run have norm bounded by $D$, by eliminating repetitions of configurations
we arrive at a run of length at most $n \cdot (D+1)^2$.
\end{proof}

%% file: images/07-disjoint-belts.tikz
\begin{tikzpicture}[scale=0.4]
    \clip (-1,-0.8) rectangle (12,16);
    \drawFadingB{12}{16}
    \drawGridB{12}{16}
    \draw[step=0.5cm,white,thin,draw opacity=0.7] (0,0) grid (3-0.01,3-0.01);
        \begin{scope}
            \clip (0,0) rectangle (12,16);
            \drawBeltB{0.5}{3}{1.5cm}{12}{16}
        \end{scope}
        \draw[style=vector] (0,0) -- node[auto,swap,pos=1.35,xshift=-0.25cm] {$u$} (0.5,3);
        \begin{scope}
            \clip (0,0) rectangle (12,16);
            \drawBeltB{1.5}{3}{1.5cm}{12}{16}
        \end{scope}
        \draw[style=vector] (0,0) -- node[auto,swap,pos=1.35,xshift=-0.25cm] {$v$} (1.5,3);

    \draw[black,thick,dashed] (-0.1,3) -- (3,3) -- (3,-0.1);
    \node at (3,-0.5) {$A$};
    \draw[black,thick,dashed] (-0.1,11) -- (11,11) -- (11,-0.1);
    \node at (11,-0.5) {$B$};
    \node at (7,9) {$\ell$};
    \draw[black,dotted] (0,0) -- (12,12);
    \draw[black,dotted] (-0.1,9) -- (12,9);
    \node[anchor=east] at (10.7,1) {\textsf{square} $S$};

    \node[point,label=above:{$p_u$}] at (3/4*2,3/4*12) {};
    \node[point,label=above:{$p_v$}] at (3/4*6,3/4*12) {};
    \node[point,label=above:{$w'$}] at (2.95,3/4*12) {};
    \node at (1-0.5,6) {$\ell_u$};
    \node at (3-0.5,6) {$\ell_v$};
    \node at (6-0.5,6) {$\ell_I$};


    \coordinate (width1) at (1.1*6,1.1*12);
    \coordinate (width2) at ($(width1)!1.5cm!90:(0,0)$);
    \draw[basicArrowBoth] (width1) -- node[auto,midway] {$W$} (width2);
\end{tikzpicture}